\documentclass[final,leqno]{siamltex}

\newcommand{\rd}{{\mathbb{R}^d}}
\newcommand{\Exp}{\mathbb E}
\newcommand{\D}{\mathcal D}

\usepackage{mathtools}
\usepackage{amssymb}
\usepackage{stmaryrd}
\usepackage{color}
\usepackage{pdflscape}
\usepackage{epsfig}

\usepackage{epstopdf}
\usepackage{multirow}
\usepackage{hhline}
\usepackage{txfonts}
\usepackage{mathrsfs}
\usepackage{color}

\usepackage{subcaption}

\usepackage{mathrsfs}
\newtheorem{thm}{Theorem}[section]

\newtheorem{rem}[thm]{Remark}

\numberwithin{equation}{section}
\title{A Unified Spectral Method for FPDEs with Two-sided Derivatives; A Fast Solver} 

\author{
	Mehdi Samiee
	\footnote{D\lowercase{epartment of} C\lowercase{omputational} M\lowercase{athematics}, S\lowercase{cience}, \lowercase{and}, E\lowercase{ngineering} \& D\lowercase{epartment of} M\lowercase{echanical} E\lowercase{ngineering},	
		M\lowercase{ichigan} S\lowercase{tate} U\lowercase{niversity}, 428 S S\lowercase{haw} L\lowercase{ane}, E\lowercase{ast} L\lowercase{ansing}, MI 48824, USA}
	, Mohsen Zayernouri
	\footnote{D\lowercase{epartment of} C\lowercase{omputational} M\lowercase{athematics}, S\lowercase{cience}, \lowercase{and}, E\lowercase{ngineering} \&
		D\lowercase{epartment of} M\lowercase{echanical} E\lowercase{ngineering},	
		M\lowercase{ichigan} S\lowercase{tate} U\lowercase{niversity}, 428 S S\lowercase{haw} L\lowercase{ane}, E\lowercase{ast} L\lowercase{ansing}, MI 48824, USA,  C\lowercase{orresponding author; zayern@msu.edu}}
	AND Mark M. Meerschaert
	\footnote{D\lowercase{epartment of} S\lowercase{tatistics and} P\lowercase{robability}, M\lowercase{ichigan} S\lowercase{state} U\lowercase{niversity}, 619 R\lowercase{ed} C\lowercase{edar} R\lowercase{oad} W\lowercase{ells} H\lowercase{all}, E\lowercase{ast} L\lowercase{ansing}, MI 48824, USA}
}

\begin{document}

\maketitle

\begin{abstract}
%
%Fractional advection-dispersion problem is one of the most interesting models for anomalous transport which applies to a range of applications such as groundwater hydrology. 
We develop a unified Petrov-Galerkin spectral method for a class of fractional partial differential equations with two-sided derivatives and constant coefficients of the form 
$ \prescript{}{0}{\mathcal{D}}_{t}^{2\tau} u^{}  
 +  
\sum_{i=1}^{d} 
[c_{l_i}\prescript{}{a_i}{\mathcal{D}}_{x_i}^{2\mu_i} u^{} +c_{r_i}\prescript{}{x_i}{\mathcal{D}}_{b_i}^{2\mu_i} u^{} ]
+ \gamma\,\, u^{}  
= \sum_{j=1}^{d} [
\kappa_{l_j} \prescript{}{a_j}{\mathcal{D}}_{x_j}^{2\nu_j} u^{} +\kappa_{r_j}\prescript{}{x_j}{\mathcal{D}}_{b_j}^{2\nu_j} u^{} ] + f$, where $2\tau \in (0,2)$, $2\mu_i \in (0,1)$ and $2\nu_j \in (1,2)$, in a ($1+d$)-dimensional \textit{space-time} hypercube, $d = 1, 2, 3, \cdots$, subject to homogeneous Dirichlet initial/boundary conditions. We employ the eigenfunctions of the fractional Sturm-Liouville eigen-problems of the first kind in \cite{zayernouri2013fractional}, called \textit{Jacobi poly-fractonomial}s, as temporal bases, and the eigen-functions of the boundary-value problem of the second kind as temporal test functions. Next, we construct our spatial basis/test functions using Legendre polynomials, yielding mass matrices being independent of the spatial fractional orders ($\mu_i, \, \nu_j, \, i, \,j=1,2,\cdots,d$). Furthermore, we formulate a novel unified fast linear solver for the resulting high-dimensional linear system based on the solution of generalized eigen-problem of spatial mass matrices with respect to the corresponding stiffness matrices, hence, making the complexity of the problem optimal, i.e., $\mathcal{O}(N^{d+2})$. We carry out several numerical test cases to examine the CPU time and convergence rate of the method. The corresponding stability and error analysis of the Petrov-Galerkin method are carried out in \cite{samiee2016Unified2}.
\end{abstract}

\pagestyle{myheadings}
\thispagestyle{plain}

\section{Introduction}
\label{Sec: Intro}
Fractional calculus seamlessly generalizes the notion of standard integer-order calculus to its fractional-order counterpart, leading to a broader class of mathematical models, namely fractional ordinary differential equations (FODEs) and fractional partial differential equations (FPDEs) \cite{podlubny1998fractional, meerschaert2012stochastic, guo2015fractional, samko1993fractional, carpinteri2014fractals}. Such non-local models appear as tractable mathematical tools to describe anomalous transport, which manifests in memory-effects, non-local interactions, power-law distributions, sharp peaks, and self-similar structures \cite{Klages2008, meerschaert2012stochastic, metzler2000random, zaslavsky1999physics}. Although anomalous, such phenomena are observed in a range of applications e.g., bioengineering \cite{perdikaris2014fractional, magin2006fractional, regner2014randomness, naghibolhosseini2015estimation}, turbulent flows \cite{solomon1993observation, solomon1994chaotic, meerschaert2014tempered, del2004fractional, del1993chaotic, zayernouri2016fractional}, porous media \cite{benson2001fractional, Benson2000, vafai2015handbook}, viscoelastic materials \cite{mainardi2010fractional}.  

Due to their history dependence and non-local character, the discretization of such problems becomes computationally challenging. Numerical methods, developed to discretize FPDEs, can be categorized in two major classes: i) local methods, e.g., finite difference method (FDM), finite volume method (FVM), and finite element method (FEM), and ii) global methods, e.g., single and multi-domain spectral methods (SM).

Local schemes have been studied extensively in the literature. Lubich introduced the discretized fractional calculus within the spirit of FDM \cite{lubich1986discretized}. Sugimoto employed a FDM scheme for approximating fractional Burger's equation \cite{sugimoto1991burgers, sugimoto1989generalized}. Meerschaert and Tadjeran \cite{meerschaert2004finite} developed finite difference approximations to solve one-dimensional advection-dispersion equations with variable coefficients on a finite domain. Tadjeran and Meerschaert \cite{tadjeran2007second} employed a practical alternating directions implicit (ADI) method to solve a class of fractional partial differential equations with variable coefficients in bounded domain. Hejazi et al. \cite{hejazi2013finite} developed a finite-volume method utilizing fractionally shifted grunwald formula for the fractional derivatives for space-fractional advection-dispersion equation on a finite domain. To solve the two-dimensional two-sided space-fractional convection diffusion equation, Chen and Deng \cite{chen2014second} proposed a practical alternating directions implicit method. Zeng et al., \cite{zeng2015numerical} constructed a finite element method and a multistep method for unconditionally stable time-integration of sub-diffusion problem. In addition, Zhao et al. developed second-order FDM for the variable-order FPDEs in \cite{zhao2015second}. Li et al. \cite{li2016linear} proposed an implicit finite difference scheme for solving the generalized time-fractional Burger's equation. Recently, Feng et al. \cite{feng2016high} proposed a second-order Crank-Nicolson scheme to approximate the Riesz space-fractional advection-dispersion equations (FADE).  Moreover, two compact non-ADI FDMs have been proposed for the high-dimensional time-fractional sub-diffusion equation by Zeng et al. \cite{zeng2016fast}. Recently, Zayernouri and Matzavinos \cite{zayernouri2016fractional} have developed an explicit fractional adams/Bashforth/Moulton and implicit fractional Adams-Moulton finite difference methods, applicable to high-order time-integration of nonlinear FPDEs and amenable for formulating implicit/explicit (IMEX) splitting methods.

Regarding global methods, Sugimoto \cite{sugimoto1991burgers, sugimoto1989generalized} used Fourier SM in a fractional Burger's equation. Shen and Wang \cite{shen2007fourierization}  constructed a set of Fourier-like basis functions for Legendre-Galerkin method for non-periodic boundary value problems and proposed a new space-time spectral method. Sweilam et al. \cite{sweilam2014chebyshev} considered Chebyshev Pseudo-spectral method for solving one-dimensional FADE, where the fractional derivative is described in Caputo sense. Chen et al. \cite{chen2015multi} developed an approach for high-order time integration within multi-domain setting for time-fractional diffusion equations. Mokhtary developed a fully discrete Galerkin method to numerically approximate initial value fractional integro-differential equations \cite{mokhtary2015discrete}. 

Moreover, Zayernouri and Karniadakis \cite{zayernouri2013fractional,zayernouri2015tempered} introduced a new family of basis/test functions, called \textit{(tempered)Jacobi poly-fractonomials}, known as the explicit eigenfunctions of (tempered) fractional Strum-Liouville problems in bounded domains of the first and second kind. Following this new spectral theory, they have developed a number of single- and multi-domain spectral methods \cite{Zayernouri_FDDEs_2013, zayernouri2014exponentially,  zayernouri2014fractional, zayernouri2015unified, zayernouri20}. Recently, Dehghan et al. \cite{dehghan2016analysis}, employed a Galerkin finite element and interpolating element free Galerkin methods for full discretization of the fractional diffusion-wave equation. They \cite{dehghan2016use} also introduced a full discretization of time-fractional diffusion and wave equations using meshless Galerkin method based on radial basis functions. Zaho et al., \cite{EPFL-ARTICLE-217895} developed a spectral method for the tempered fractional diffusion equations (TFDEs) using the generalized Jacobian functio \cite{chen2015generalized}. Mao and Shen \cite{mao2016efficient} developed Galerkin spectral methods for solving multi-dimensional fractional elliptic equations with variable coefficients. Besides, Lischke et al. \cite{Lischke2017petrov} presented a tunably accurate Laguerre Petrov-Galerkin spectral method for solving linear multiterm fractional initial value problems. Kharazmi et al. \cite{kharazmi2016petrov} developed a new Petrov-Galerkin spectral element method for one-dimensional fractional elliptic problems using the standard spectral element bases and the Jacobi poly-fractonomials as the test functions. 

%Recently, Zayernouri et. al \cite{zayernouri2015unified} developed a unified spectral method for FPDEs, employing a new family of fractional bases, called Jacobi poly-fractonomials \cite{zayernouri2013fractional}. Based on the choice of the basis and test functions, they formulated a fast solver allowing to solve high-dimensional advection/diffusion/wave equations. 

%
%explicitly given as:
%\begin{equation}
%\label{Eq: SingularBasis I}
%%
% \prescript{(1)}{}{ \mathcal{P}}_{n}^{\alpha, \beta,\mu}(\xi) = (1+\xi)^{-\beta+\mu-1} P_{n-1}^{\alpha-\mu+1, -\beta+\mu-1} (\xi),\,\,\,\,\,\xi\in [-1,1], 
%%
%\end{equation}
%%
%with $\mu \in (0,1)$, $-1 \le \alpha< 2 -\mu$, and $-1 \le \beta < \mu -1$, , and
%%
%\begin{equation}
%\label{Eq: SingularTest II}
%%
% \prescript{(2)}{}{ \mathcal{P}}_{n}^{\alpha, \beta,\mu}(\xi) = (1-\xi)^{-\alpha+\mu-1} P_{n-1}^{-\alpha+\mu-1\,,\, \beta-\mu+1} (\xi),\,\,\,\,\,\xi\in [-1,1], 
%%
%\end{equation}
%where $-1<\alpha<\mu-1$ and $-1<\beta<2-\mu$, and $\mu \in (0,1)$, denoting . 
%Also, they developed a Petrov-Galerkin method for solving the tempered fractional ODEs [.............]. 
The main contribution of the present work is to construct a unified Petrov-Galerkin spectral method and a unified fast solver for the weak form of linear FPDEs with constant coefficients in (1+d) dimensional \textit{space-time} hypercube of the form 
\begin{eqnarray}
\label{111111}
\nonumber
\prescript{}{0}{\mathcal{D}}_{t}^{2\tau} u^{}  
 +  
\sum_{i=1}^{d} 
[c_{l_i}\prescript{}{a_i}{\mathcal{D}}_{x_i}^{2\mu_i} u^{} +c_{r_i}\prescript{}{x_i}{\mathcal{D}}_{b_i}^{2\mu_i} u^{} ]
&=&  \sum_{j=1}^{d} 
[\kappa_{l_j}\prescript{}{a_j}{\mathcal{D}}_{x_j}^{2\nu_j} u^{} +\kappa_{r_j}\prescript{}{x_j}{\mathcal{D}}_{b_j}^{2\nu_j} u^{} ]
\\
&&
% \nonumber
-\gamma\,\, u^{}  +
f,
\end{eqnarray}
where $2\mu_i, \, \in [0, \, 1]$, $2\nu_i, \, \in [1, \, 2]$, and $2\tau, \, \in [0, \, 2]$ subject to Dirichlet initial and boundary conditions, where $i=1, \, 2, \, ..., \, d$. Compared to the problem considered in [45], we extend the one-sided spatial derivatives to two-sided ones, also, we include an advection term in order to consider the drift effects. Employing different (Legendre polynomial) spatial basis/test functions and the additional advection term then would not allow employing the fast linear solver developed in [45]. Accordingly, we formulate a new fast linear solver for advection-dispersion problems. We additionally aim to perform the inf-sup stability analysis in any (1+d) dimensions in [2], while in [45], only the stability of 1-D problem has been carried out. Furthermore, we briefly presented the stochastic interpretation of FADE on bounded domain which sheds light on the well-posedness of the problem from the perspective of the probability theory. In \cite{samiee2016Unified2}, we also carry out the corresponding error analyses of the PG method along with several verifying numerical tests.  

The outline of this paper is as follows: in section \ref{Sec: Notation}, we introduce some preliminary results from fractional calculus. In section \ref{Sec: General FPDE}, we present the mathematical formulation of the spectral method in a (1+d) dimensional space, which leads to the generalized Lyapunov equations. In section \ref{Sec: FastSolver FPDE}, we develop a unified fast linear solver and obtain the closed-form solution in terms of the genralized eigenvalues and eigenvectors of the corresponding mass and stiffness matrices. In section \ref{num-test}, the performance of the PG method is examined via several numerical simulations for low-to- high dimensional problems with smooth and non-smooth solutions.

%%%%%%%%%%%%%%%%%%%%%%%%%%%%%%%%%%
%
\section{Preliminaries on Fractional Calculus}
\label{Sec: Notation}
%
%%%%%%%%%%%%%%%%%%%%%%%%%%%%%%%%%%
%
Here, we obtain some basic definitions from fractional calculus \cite{meerschaert2012stochastic, zayernouri2015unified}. Denoted by $\prescript{}{a}{\mathcal{D}}_{x}^{\nu} g(x)$, the left-sided Reimann-Liouville fractional derivative of order $\nu$ in which $g(x) \in C^{n}[a,b]$ and $n= \lceil \nu \rceil$, is defined as:
\begin{equation}
\label{eq2}
\prescript{RL}{a}{\mathcal{D}}_{x}^{\nu} g(x) = \frac{1}{\Gamma(n-\nu)}  \frac{d^{n}}{d x^n} \int_{a}^{x} \frac{g(s) }{(x - s)^{\nu +1-n} }\,\,ds,\quad x \in [a,b],
\end{equation}
where $\Gamma$ represents the Euler gamma function. The corresponding right-sided Reimann-Liouville fractional derivative of order $\nu$, $\prescript{}{x}{\mathcal{D}}_{b}^{\nu} g(x)$, is given by
\begin{equation}
\label{eq3}
\prescript{RL}{x}{\mathcal{D}}_{b}^{\nu} g(x) = \frac{1}{\Gamma(n-\nu)}  (-1)^{n} \frac{d^{n}}{d x^n} \int_{x}^{b} \frac{g(s) }{(s - x)^{\nu +1-n} }\,\,ds,\quad x \in [a,b].
\end{equation}
In \eqref{eq2} and \eqref{eq3}, as $\nu \rightarrow n$, the fractional derivatives tend to the standard $n$-th order derivative with respect to $x$. We recall from \cite{zayernouri2013fractional,Askey1969Integral} that the following link between the Reimann-Liouville and Caputo fractional derivatives, where
\begin{eqnarray}
\label{eq4}
\prescript{Rl}{a}{\mathcal{D}}_{x}^{\nu} f(x) \,=\, 
\frac{f ( a ) }{\Gamma(1-\mu)(x-a)^{\nu}}\,+\,
\prescript{C}{a}{\mathcal{D}}_{x}^{\nu} f(x) 
\\
\label{eq5}
\prescript{Rl}{x}{\mathcal{D}}_{b}^{\nu} f(x) \,=\, 
\frac{f ( b ) }{\Gamma(1-\mu)(b-x)^{\nu}}\,+\,
\prescript{C}{x}{\mathcal{D}}_{b}^{\nu} f(x),
\end{eqnarray}
where
\begin{eqnarray}
\label{eq6}
\prescript{C}{a}{\mathcal{D}}_{x}^{\nu} f(x) \,=\, 
\frac{1}{\Gamma(n-\nu)} \int_{a}^{x} \frac{g^{(n)}(s) }{(x - s)^{\nu +1-n} }\,\,ds,\quad x \in [a,b], 
\\
\label{eq7}
\prescript{C}{x}{\mathcal{D}}_{b}^{\nu} f(x) \,=\, 
\frac{(-1)^n}{\Gamma(n-\nu)} \int_{x}^{b} \frac{g^{(n)}(s) }{(x - s)^{\nu +1-n} }\,\,ds,\quad x \in [a,b].
\end{eqnarray}
In \eqref{eq4} and \eqref{eq5}, $\prescript{RL}{a}{\mathcal{D}}_{x}^{\nu} g(x)=\prescript{C}{a}{\mathcal{D}}_{x}^{\nu} g(x)=\prescript{}{a}{\mathcal{D}}_{x}^{\nu} g(x)$ when homogeneous Dirichlet initial and boundary conditions are enforced. 

To analytically obtain the fractional differentiation of our basis function, we employ the following relations \cite{zayernouri2013fractional} as:  

\begin{eqnarray}
\label{6}
\prescript{RL}{-1}{\mathcal{I}}_{x}^{\nu} \{ (1+x)^{\beta} P_{n}^{\alpha, \beta} {(x)}\} = \frac{\Gamma(n+\beta+1)}{\Gamma(n+\beta+\nu+1)}\, (1+x)^{\beta+\nu} P_{n}^{\alpha-\nu,\beta+\nu}{(x)},
\end{eqnarray}
and
\begin{eqnarray}
\label{7}
\prescript{RL}{x}{\mathcal{I}}_{1}^{\nu} \{ (1-x)^{\alpha} P_{n}^{\alpha, \beta} {(x)}\} = \frac{\Gamma(n+\alpha+1)}{\Gamma(n+\alpha+\nu+1)}\, (1-x)^{\alpha+\nu} P_{n}^{\alpha+\nu,\beta-\nu}{(x)},
\end{eqnarray}
where $0 < \nu < 1$, $\alpha > -1$, $\beta > -1$ and $P^{\alpha, \, \beta}_{n} (x)$ denote the standard Jacobi Polynomials of order n and parameters $\alpha$ and $\beta$. It is worth mentioning that 
\begin{eqnarray}
\nonumber
\prescript{RL}{a}{\mathcal{I}}_{x}^{\nu} \{ f (x)\} = \frac{1}{\Gamma(\nu)}\, \int_{a}^{x} \frac{f(s) }{(x - s)^{1-\nu} }\,\,ds,\quad x \in [a,b], \, \,
\end{eqnarray}
 and
 \begin{eqnarray}
\prescript{RL}{x}{\mathcal{I}}_{b}^{\nu} \{ f (x)\} = \frac{1}{\Gamma(\nu)}\, \int_{x}^{b} \frac{f(s) }{(s-x)^{1-\nu} }\,\,ds,\quad x \in [a,b].
\nonumber
\end{eqnarray}
By substituting $\alpha = +\nu$ and $\beta = -\nu$, we can simplify equations (\ref{6}) and (\ref{7}), thereby we have:
\begin{eqnarray}
\label{10}
\prescript{RL}{-1}{\mathcal{I}}_{x}^{\nu} \{ (1+x)^{-\nu} P_{n}^{\nu, -\nu} {(x)}\} = \frac{\Gamma(n-\nu+1)}{\Gamma(n+1)}P_{n}{(x)},  \quad x \in [-1,1]
\end{eqnarray}
and
\begin{eqnarray}
\label{11}
\prescript{RL}{x}{\mathcal{I}}_{1}^{\nu} \{ (1-x)^{-\nu} P_{n}^{-\nu, \nu} {(x)}\} = \frac{\Gamma(n-\nu+1)}{\Gamma(n+1)}P_{n}{(x)}  ,\quad x \in [-1,1].
\end{eqnarray}
Accordingly, we have the fractional derivative of Legendre polynomial by differentiating \eqref{10} and \eqref{11} as
\begin{eqnarray}
\label{Eq: 10}
\prescript{}{-1}{\mathcal{D}}_{x}^{\nu} P_{n} (x)= \frac{\Gamma(n+1)}{\Gamma(n-\nu+1)}P_{n}^{\, \nu,-\nu}{(x)}\,(1+x)^{-\nu}
\end{eqnarray}
and
\begin{eqnarray}
\label{Eq: 11}
\prescript{}{x}{\mathcal{D}}_{1}^{\nu} P_{n} (x)= \frac{\Gamma(n+1)}{\Gamma(n-\nu+1)}P_{n}^{\, -\nu,\nu}{(x)}\,(1-x)^{-\nu}, \,
\end{eqnarray}
where $P_{n} (x) = P^{\, 0,0}_{n} (x)$ represents Legendre polynomial of degree n.

\section{Mathematical Framework}
\label{Sec: General FPDE}
%
%%%%%%%%%%%%%%%%%%%%%%%%%%%%%%%%%%%%%%%%%%%%
Let $u: \mathbb{R}^{d+1} \rightarrow \mathbb{R}$ for some positive integer $d$ and $\Omega =  [0,T]\times [a_1,b_1]\times [a_2,b_2] \times \cdots \times [a_d,b_d]$, where 
%(see e.g., \eqref{Eq: U2}) 
\begin{eqnarray}
\label{Eq: Genral_FPDE}
%
%\nonumber
 \prescript{}{0}{\mathcal{D}}_{t}^{2\tau} u  &+&  
\sum_{i=1}^{d} 
\big{[}c_{l_i}\prescript{}{a_i}{\mathcal{D}}_{x_i}^{2\mu_i} u+c_{r_i}\prescript{}{x_i}{\mathcal{D}}_{b_i}^{2\mu_i} u\big{]}
%\nonumber
%\\
%&-& 
-\sum_{j=1}^{d} 
\big{[}\kappa_{l_j}\prescript{}{a_j}{\mathcal{D}}_{x_j}^{2\nu_j} u+\kappa_{r_j}\prescript{}{x_j}{\mathcal{D}}_{b_j}^{2\nu_j} u\big{]}
+\gamma\,\, u = f,
\end{eqnarray} 
and $\gamma, c_{l_i}, \, c_{r_i}, \, \kappa_{l_j},$ and  $\kappa_{r_j}$ are all constant. Besides, $2\mu_i \in (0,1)$, $2\nu_j \in (1,2)$, and $2\tau \in (0,2)$, for $j=1,2,\cdots,d$.
This equation is subject to the following Dirichlet initial and boundary conditions as:
\begin{eqnarray}
\label{Eq: Genral_FPDE Initial/BCs}
 \nonumber
 u|_{t=0} &=& 0, \quad \tau \in (0,1/2),
 \\ \nonumber 
  u|_{t=0} &=& \frac{\partial u}{\partial t}|_{t=0} = 0, \quad \tau \in (1/2,1),
 \\ \nonumber 
 u|_{x_j=a_j}= u|_{x_j=b_j}&=&0, \quad \nu_j \in (1/2,1),\quad j=1,2,\cdots, d.
\end{eqnarray}

\subsection{\textbf{Stochastic Interpretation of the FPDEs}}
\label{statistic interp}
Following \cite{baeumer2016space}, we provide a brief stochastic interpretation of the FPDEs in \eqref{Eq: Genral_FPDE} that further sheds light on the well-posedness of the problem from the perspective of probability theory. Let suppose that in \eqref{Eq: Genral_FPDE}, $f\equiv 0$ and $\gamma=0$ and $0<2\tau<1$ and that $a_i=-\infty$ and $b_i=+\infty$ for $i=1,2,\cdots$.  Then \eqref{Eq: Genral_FPDE} governs \cite{baeumer2016space} a time-changed L\'evy process $X(E_t)$ on $\rd$ whose Fourier transform is
% where $E_t=\inf\{x>0:D_x>t\}$ is the inverse of the stable subordinator $D_t$ with $\Exp[e^{-sD_t}]=e^{-ts^{2\tau}}$ and the L\'evy process has Fourier transform  
 $\Exp[e^{-ik\cdot X(t)}]=e^{t\psi(k)}$ with the Fourier symbol
\begin{equation}
\begin{split}\label{eq22}
\psi(k)= -\sum_{n=1}^d [c_{l_n} (ik_n)^{2\mu_n}&+c_{r_n} (-ik_n)^{2\mu_n}]
+\sum_{m=1}^d [\kappa_{l_m}(-ik_m)^{2\nu_m}+\kappa_{r_m}(ik_m)^{2\nu_m}].
\end{split}
\end{equation}
Recalling that in one dimension the L\'evy process $Y(t)$ with Fourier Transform $\Exp[e^{-ikY(t)}]=e^{t\psi_0(k)}$ where $\psi_0(k)=pD(ik)^\alpha+qD(-ik)^\alpha$ for $D>0$ and $1<\alpha\leq 2$, $p\geq 0$, $q\geq 0$, and $p+q=1$ is a stable L\'evy process with index $\alpha$ and skewness $p-q$ \cite{meerschaert2012stochastic,baeumer2016space}.  
In brief, fractional advection-dispersion equation on unbounded domain is represented by a solution involves an inverse stable subordinator time-changed, resulting in an non-Markovian process. You can find complete details in  \cite{meerschaert2012stochastic}.
\vspace{0.02cm}

Regarding a computational domain, Chen et al. \cite{chen2012space} developed a solution for the case of equation \eqref{Eq: Genral_FPDE} where $f\equiv 0$ and $\gamma=0$ and $0<2\tau<1$ and all $a_m=a_n>-\infty$ and $b_m=b_n<\infty$, with zero Dirichlet boundary conditions. It follows from \cite{baeumer2016space} that
\begin{equation}\label{gen1}
Lu(x) =c u'(x)+\kappa_{l}\, _{a}\D_{x}^{2\nu}u(x)+\kappa_{r}\, _{x}\D_{b}^{2\nu}u(x)
\end{equation}
is the generator of the killed semigroup on the bounded domain $\Omega=(a,b)$ which is also the point source to \eqref{Eq: Genral_FPDE}.  In other words, starting with the point source initial condition $u(x,0)=\delta(x)$, the solution to \eqref{Eq: Genral_FPDE} with the restrictions discussed in \cite{meerschaert2012stochastic,chen2012space} is the PDF of a \textit{killed non-Markovian} process.

\subsection{\textbf{Mathematical Framework}}
In \cite{li2010existence}, the usual Sobolev space associated with the real index $\sigma \geq 0$ on bounded interval $\Lambda = (a,b)$, is denoted by $H^{\sigma}_{}(\Lambda)$ and is defined as the completion of $C_0^{\infty}(\Lambda)$ with respect to the norm $\Vert \cdot \Vert_{H^{\sigma}_{}(\Lambda)}$. As shown in Lemma 2.6 in \cite{li2010existence}, the equivalency between the following norms holds:
\begin{equation}
\Vert \cdot \Vert_{H^{\sigma}_{}(\Lambda)} \equiv \Vert \cdot \Vert_{{^l}H^{\sigma}_{}(\Lambda)} \equiv \Vert \cdot \Vert_{{^r}H^{\sigma}_{}(\Lambda)},
\end{equation}
where 
\begin{equation}
\Vert \cdot \Vert_{{^l}H^{\sigma}_{}(\Lambda)} = \Big(\Vert \prescript{}{a}{\mathcal{D}}_{x}^{\sigma}\, (\cdot)\Vert_{L^2(\Lambda)}^2+\Vert \cdot \Vert_{L^2(\Lambda)}^2 \Big)^{\frac{1}{2}},
\end{equation}
and 
\begin{equation}
\Vert \cdot \Vert_{{^r}H^{\sigma}_{}(\Lambda)} = \Big(\Vert \prescript{}{x}{\mathcal{D}}_{b}^{\sigma}\, (\cdot)\Vert_{L^2(\Lambda)}^2+\Vert \cdot \Vert_{L^2(\Lambda)}^2 \Big)^{\frac{1}{2}}.
\end{equation}
Similarly, we can show that $\Vert \cdot \Vert_{H^{\sigma}_{}(\Lambda)} \equiv \Vert \cdot \Vert_{{^c}H^{\sigma}_{}(\Lambda)}$, defined as
\begin{equation}
\Vert \cdot \Vert_{{^c}H^{\sigma}_{}(\Lambda)} = \Big(\Vert \prescript{}{x}{\mathcal{D}}_{b}^{\sigma}\, (\cdot)\Vert_{L^2(\Lambda)}^2+\Vert \prescript{}{a}{\mathcal{D}}_{x}^{\sigma}\, (\cdot)\Vert_{L^2(\Lambda)}^2+\Vert \cdot \Vert_{L^2(\Lambda)}^2 \Big)^{\frac{1}{2}}.
\end{equation} 
Let $\Lambda_1 = (a_1,b_1)$, $\Lambda_i = (a_i,b_i) \times \Lambda_{i-1}$ for $i=2,\cdots,d$, and $\mathcal{X}_1 = H^{\nu_1}_{0}(\Lambda_1)$, with the associated norm $\Vert \cdot \Vert_{{}H^{\nu_1}_{}(\Lambda_1)} \equiv \Vert \cdot \Vert_{{^c}H^{\nu_1}_{}(\Lambda_1)}$. Accordingly, we construct $\mathcal{X}_d$ such that
\begin{eqnarray}
\mathcal{X}_2 &=& H^{\nu_2}_0 \Big((a_2,b_2); L^2(\Lambda_1) \Big) \cap L^2((a_2,b_2); \mathcal{X}_1),
\nonumber
\\
&\vdots&
\nonumber
\\
\mathcal{X}_d &=& H^{\nu_d}_0 \Big((a_d,b_d); L^2(\Lambda_{d-1}) \Big) \cap L^2((a_2,b_2); \mathcal{X}_{d-1}),
\end{eqnarray}
associated with the norm
\begin{equation}
\label{norm_Xd}
\Vert \cdot \Vert_{\mathcal{X}_d} = \bigg{\{} \Vert \cdot \Vert_{L^2(\Lambda_d)}^2 + \sum_{i=1}^{d} \Big(\Vert \prescript{}{x_i}{\mathcal{D}}_{b_i}^{\nu_i}\, (\cdot)\Vert_{L^2(\Lambda_d)}^2+\Vert \prescript{}{a_i}{\mathcal{D}}_{x_i}^{\nu_i}\, (\cdot)\Vert_{L^2(\Lambda_d)}^2 \Big) \bigg{\}}^{\frac{1}{2}}.
\end{equation}
Similarly, the Sobolev space with index $\tau>0$ on the time interval $I=(0,T)$, denoted by $H^{\tau}_{}(I)$, is endowed with norm $\Vert \cdot \Vert_{H^{\tau}_{}(I)}$, where
\begin{equation}
\Vert \cdot \Vert_{H^{\tau}_{}(I)} \equiv \Vert \cdot \Vert_{{^l}H^{\tau}_{}(I)} \equiv \Vert \cdot \Vert_{{^r}H^{\tau}_{}(I)},
\end{equation}
\begin{equation}
\Vert \cdot \Vert_{{^l}H^{\tau}_{}(I)} = \Big(\Vert \prescript{}{0}{\mathcal{D}}_{t}^{\tau}\, (\cdot)\Vert_{L^2(I)}^2+\Vert \cdot \Vert_{L^2(I)}^2 \Big)^{\frac{1}{2}},
\end{equation}
and 
\begin{equation}
\Vert \cdot \Vert_{{^r}H^{\tau}_{}(I)} = \Big(\Vert \prescript{}{t}{\mathcal{D}}_{T}^{\tau}\, (\cdot)\Vert_{L^2(I)}^2+\Vert \cdot \Vert_{L^2(I)}^2 \Big)^{\frac{1}{2}}.
\end{equation}
Let $2\tau \in (0,1)$ and $\Omega=I \times \Lambda_d$. We define
\begin{equation}
\prescript{l}{0}H^{\tau} \Big(I; L^2(\Lambda_d) \Big) := \Big{\{} u \,|\, \Vert u(t,\cdot) \Vert_{L^2(\Lambda_d)} \in H^{\tau}(I), u\vert_{t=0}=u\vert_{x=a_i}=u\vert_{x=b_i}=0,\, i=1,\cdots,d  \Big{\}},
\end{equation}
which is equipped with the norm
\begin{eqnarray}
\Vert u \Vert_{\prescript{l}{}H^{\tau}(I; L^2(\Lambda_d))} &=& \Big{\Vert} \, \Vert u(t,\cdot) \Vert_{L^2(\Lambda_d)}\, \Big{\Vert}_{{^l}H^{\tau}(I)}
= \Big(\Vert \prescript{}{0}{\mathcal{D}}_{t}^{\tau}\, u\Vert_{L^2(\Omega)}^2+\Vert u \Vert_{L^2(\Omega)}^2 \Big)^{\frac{1}{2}}.
\end{eqnarray}
Similarly,
\begin{equation}
\prescript{r}{0}H^{\tau} \Big(I; L^2(\Lambda_d) \Big) := \Big{\{} v \,|\, \Vert v(t,\cdot) \Vert_{L^2(\Lambda_d)} \in H^{\tau}(I), v\vert_{t=T}=v\vert_{x=a_i}=v\vert_{x=b_i}=0,\, i =1,\cdots,d  \Big{\}},
\end{equation}
which is equipped with the norm
\begin{eqnarray}
\Vert v \Vert_{\prescript{r}{}H^{\tau}(I; L^2(\Lambda_d))} &=& \Big{\Vert} \, \Vert v(t,\cdot) \Vert_{L^2(\Lambda_d)}\, \Big{\Vert}_{{^r}H^{\tau}(I)}
= \Big(\Vert \prescript{}{t}{\mathcal{D}}_{T}^{\tau}\, v \Vert_{L^2(\Omega)}^2+\Vert v \Vert_{L^2(\Omega)}^2\Big)^{\frac{1}{2}}
\end{eqnarray}
We define the solution space
\begin{equation}
\mathcal{B}^{\tau,\nu_1,\cdots,\nu_d} (\Omega):= \prescript{l}{0}H^{\tau}\Big(I; L^2(\Lambda_d) \Big) \cap L^2(I; \mathcal{X}_d),
\end{equation}
endowed with the norm
\begin{equation}
\Vert u \Vert_{\mathcal{B}^{\tau,\nu_1,\cdots,\nu_d}} = \Big{\{}\Vert u \Vert_{\prescript{l}{}H^{\tau}(I; L^2(\Lambda_d))}^2 + \Vert u \Vert_{L^2(I; \mathcal{X}_d)}^2 \Big{\}}^{\frac{1}{2}},
\end{equation}
where due to \eqref{norm_Xd},
\begin{eqnarray}
\Vert u \Vert_{L^2(I; \mathcal{X}_d)}&=&  \Big{\Vert} \, \Vert u(t,.) \Vert_{\mathcal{X}_d}\,\Big{\Vert}_{L^2(I)}
\nonumber
\\
&=& \Big{\{}  \Vert u \Vert_{L^2(\Omega)}^2 + \sum_{i=1}^{d} \big( \Vert \prescript{}{x_i}{\mathcal{D}}_{b_i}^{\nu_i}\, (u)\Vert_{L^2(\Omega)}^2+\Vert \prescript{}{a_i}{\mathcal{D}}_{x_i}^{\nu_i}\, (u)\Vert_{L^2(\Omega)}^2 \big) \Big{\}}^{\frac{1}{2}}.  \quad \quad
\end{eqnarray}
Therefore,
\begin{equation}
\Vert u \Vert_{\mathcal{B}^{\tau,\nu_1,\cdots,\nu_d}} = \Big{\{}  \Vert u \Vert_{L^2(\Omega)}^2 + \Vert \prescript{}{0}{\mathcal{D}}_{t}^{\tau}\, (u)\Vert_{L^2(\Omega)}^2 + \sum_{i=1}^{d} \big( \Vert \prescript{}{x_i}{\mathcal{D}}_{b_i}^{\nu_i}\, (u)\Vert_{L^2(\Omega)}^2+\Vert \prescript{}{a_i}{\mathcal{D}}_{x_i}^{\nu_i}\, (u)\Vert_{L^2(\Omega)}^2 \big) \Big{\}}^{\frac{1}{2}}. 
\end{equation}
Likewise, we define the test space
\begin{equation}
\mathfrak{B}^{\tau,\nu_1,\cdots,\nu_d} (\Omega) := \prescript{r}{}H^{\tau}\Big(I; L^2(\Lambda_d)\Big) \cap L^2(I; \mathcal{X}_d),
\end{equation}
endowed with the norm
\begin{eqnarray}
\Vert v \Vert_{\mathfrak{B}^{\tau,\nu_1,\cdots,\nu_d}} &=& \Big{\{}\Vert v \Vert_{\prescript{r}{}H^{\tau}(I; L^2(\Lambda_d))}^2  + \Vert v \Vert_{ L^2(I; \mathcal{X}_d)}^2 \Big{\}}^{\frac{1}{2}}.
\nonumber
\\
&=& \Big{\{}  \Vert v \Vert_{L^2(\Omega)}^2 + \Vert \prescript{}{t}{\mathcal{D}}_{T}^{\tau}\, (v)\Vert_{L^2(\Omega)}^2 + \sum_{i=1}^{d} \big( \Vert \prescript{}{x_i}{\mathcal{D}}_{b_i}^{\nu_i}\, (v)\Vert_{L^2(\Omega)}^2+\Vert \prescript{}{a_i}{\mathcal{D}}_{x_i}^{\nu_i}\, (v)\Vert_{L^2(\Omega)}^2 \big) \Big{\}}^{\frac{1}{2}}. \quad \quad
\end{eqnarray}
In case $2\tau \in (1,2)$, we define the solution space as
\begin{equation}
\mathcal{B}^{\tau,\nu_1,\cdots,\nu_d}(\Omega) := \prescript{l}{0,0}H^{\tau}\Big(I; L^2(\Lambda_d)\Big) \cap L^2(I; \mathcal{X}_d),
\end{equation}
where 
\begin{eqnarray}
\prescript{l}{0,0}H^{\tau}\Big(I; L^2(\Lambda_d)\Big) &:=& \Big{\{} u \,|\, \Vert u(t,\cdot) \Vert_{L^2(\Lambda_d)} \in H^{\tau}(I), 
\nonumber
\\
&&
\frac{\partial u}{\partial t}\vert_{t=0}= u\vert_{t=0}=u\vert_{x=a_i}=u\vert_{x=b_i}=0,\, i=1,\cdots,d  \Big{\}},
\nonumber
\end{eqnarray}
which is associated with $\Vert \cdot \Vert_{\mathcal{B}^{\tau,\nu_1,\cdots,\nu_d}}$. 
The corresponding test space is also defined as
\begin{equation}
\mathfrak{B}^{\tau,\nu_1,\cdots,\nu_d}(\Omega) := \prescript{r}{0,0}H^{\tau}\Big(I; L^2(\Lambda_d)\Big) \cap L^2(I; \mathcal{X}_d),
\end{equation}
where 
\begin{eqnarray}
\prescript{r}{0,0}H^{\tau}\Big(I; L^2(\Lambda_d)\Big) &:=& \Big{\{} v \,|\, \Vert v(t,\cdot) \Vert_{L^2(\Lambda_d)} \in H^{\tau}(I),
\nonumber
\\
&&
 \frac{\partial v}{\partial t}\vert_{t=T}= v\vert_{t=T}=v\vert_{x=a_i}=v\vert_{x=b_i}=0,\, i=1,\cdots,d  \Big{\}},
 \nonumber
\end{eqnarray}
which is endowed with $\Vert \cdot \Vert_{\mathfrak{B}^{\tau,\nu_1,\cdots,\nu_d}}$. 

\subsection{\textbf{Petrov-Galerkin Method}}
\label{PGmethod}

Next, we define the corresponding bilinear form as 
\begin{eqnarray}
%\label{Eq: Genral_FPDE}
a(u,v)&=&
(\prescript{}{0}{\mathcal{D}}_{t}^{\tau}\, u, \prescript{}{t}{\mathcal{D}}_{T}^{\tau}\, v )_{\Omega}  
\nonumber
\\
&&+
\sum_{i=1}^{d} \big{[}c_{l_i}  ( \prescript{}{a_i}{\mathcal{D}}_{x_i}^{\mu_i}\, u,\, \prescript{}{x_i}{\mathcal{D}}_{b_i}^{\mu_i}\, v )_{\Omega}+ c_{r_i}  ( \prescript{}{x_i}{\mathcal{D}}_{a_i}^{\mu_i}\, u,\, \prescript{}{a_i}{\mathcal{D}}_{x_i}^{\mu_i}\, v )_{\Omega} \big{]}
\nonumber
\\ 
&&
-
\sum_{j=1}^{d} \big{[}\kappa_{l_j}  ( \prescript{}{a_j}{\mathcal{D}}_{x_j}^{\nu_j}\, u,\, \prescript{}{x_j}{\mathcal{D}}_{b_j}^{\nu_j}\, v )_{\Omega}+ \kappa_{r_j}  ( \prescript{}{x_j}{\mathcal{D}}_{b_j}^{\nu_j}\, u,\, \prescript{}{a_j}{\mathcal{D}}_{x_j}^{\nu_j}\, v )_{\Omega} \big{]}
\nonumber
\\ 
&&
+\gamma 
(u,v)_{\Omega}.
\end{eqnarray}
Now, the problem reads as: find $u \in \mathcal{B}^{\tau,\nu_1,\cdots,\nu_d}(\Omega)$ such that 
\begin{eqnarray}
\label{Eq: con-dim PG method}
&a(u,v) = (f,v)_{\Omega}, \quad \forall v \in \mathfrak{B}^{\tau,\nu_1,\cdots,\nu_d}(\Omega),&
\end{eqnarray} 
where $a(u,v)$ is a continuous bilinear form and $f \in {\mathbb{B}}^{\tau,\nu_1,\cdots,\nu_d}(\Omega)$, which is the dual space of $\mathcal{{B}}^{\tau,\nu_1,\cdots,\nu_d}(\Omega)$. It should be noted that $(\prescript{}{0}{\mathcal{D}}_{t}^{2\tau}\, u, \, v )_{\Omega}=(\prescript{}{0}{\mathcal{D}}_{t}^{\tau}\, u, \prescript{}{t}{\mathcal{D}}_{T}^{\tau}\, v )_{\Omega}$ is proven in Lemma 4 in \cite{zhang2010galerkin} and later in \cite{kharazmi2017petrov} requiring less regularity and constraint. Therefore, We construct a Petrov-Galerkin spectral method for $u \in \mathcal{{B}}^{\tau,\nu_1,\cdots,\nu_d}(\Omega)$, satisfying the weak form of (\ref{Eq: Genral_FPDE}) as
\begin{eqnarray}
\label{Eq: general weak form}
\nonumber
%\prescript{}{0}{\mathcal{D}}_{t}^{2\tau} u +\sum_{j=1}^{d} c_{x_j}\prescript{}{a_j}{\mathcal{D}}_{x_j}^{2\mu_j} u +\gamma\,\, u=f,
(\prescript{}{0}{\mathcal{D}}_{t}^{\tau}\, u, \prescript{}{t}{\mathcal{D}}_{T}^{\tau}\, v )_{\Omega}  &+&
\sum_{i=1}^{d} \big{[}c_{l_i}  ( \prescript{}{a_i}{\mathcal{D}}_{x_i}^{\mu_i}\, u,\, \prescript{}{x_i}{\mathcal{D}}_{b_i}^{\mu_i}\, v )_{\Omega} 
+ c_{r_i}  ( \prescript{}{a_i}{\mathcal{D}}_{x_i}^{\mu_i}\, v,\, \prescript{}{x_i}{\mathcal{D}}_{b_i}^{\mu_i}\, u )_{\Omega}\big{]} 
\\
&-&\sum_{j=1}^{d} [k_{l_j}  ( \prescript{}{a_j}{\mathcal{D}}_{x_j}^{\nu_j}\, u,\, \prescript{}{x_j}{\mathcal{D}}_{b_j}^{\nu_j}\, v )_{\Omega}+k_{r_j}  ( \prescript{}{a_j}{\mathcal{D}}_{x_j}^{\nu_j}\, v,\, \prescript{}{x_j}{\mathcal{D}}_{b_j}^{\nu_j}\, u )_{\Omega}] 
\nonumber
\\
&+& \gamma 
(u,v)_{\Omega} = (f,v)_{\Omega}, \quad \forall v \in \mathfrak{B}^{\tau,\nu_1,\cdots,\nu_d}(\Omega), 
\end{eqnarray}
where $(\cdot,\cdot)_{\Omega}$ represents the usual $L^2$-product. 
%is equal to the variational form \eqref{Eq: general weak form} when solution $u$ is smooth enough. 

Next, we choose proper subspaces of $\mathcal{B}^{\tau,\nu_1,\cdots,\nu_d}(\Omega)$ and $\mathfrak{B}^{\tau,\nu_1,\cdots,\nu_d}(\Omega)$ as finite dimensional $U_N$ and $V_N$ with $\dim(U_N) = \dim(V_N) =N$. Now, the discrete problem reads: find $u_N\in U_N$ such that
\begin{eqnarray}
\label{Eq: Finite-dim PG method}
&a(u_N,v_N) = (f,v_N), \quad \forall v_N \in V_N.&
\end{eqnarray} 
By representing $u_N$ as a linear combination of points/elements in $U_N$, i.e., the corresponding $(1+d)$-dimensional space-time basis functions, the finite-dimensional problem \eqref{Eq: Finite-dim PG method} leads to a linear system known as \textit{Lyapunov} system. For instance, when $d=1$,  we obtain the corresponding Lyapunov equation in the space-time domain $(0,T)\times(a_1,b_1)$ as 
\begin{eqnarray}
\label{Eq: total Lyapunov-2d1}
S_{\tau}\,\mathcal{U}\,M_1^T &+& c_{l_1} M_{\tau}\,\mathcal{U}\, S_{\mu_1,l}^T+ c_{r_1} M_{\tau}\,\mathcal{U}\, S_{\mu_1,r}^T 
\nonumber
\\
&-& \kappa_{l_1} M_{\tau}\,\mathcal{U}\, S_{\nu_1,l}^T- \kappa_{r_1} M_{\tau}\,\mathcal{U}\, S_{\nu_1,r}^T + \gamma M_{\tau}\,\mathcal{U}\,M_1^T =F, \quad
\end{eqnarray}
where all are defined in \ref{Sec: Implementaiton of PG}.
To find the general form of Lyapunov equation, we can define $S^{Tot}$ as
\begin{eqnarray}
\label{Eq: general Lyapunov-2d}
& -\kappa_{l_1} \,S_{\nu_1,l} - \kappa_{r_1} \,S_{\nu_1,r} + c_{l_1} \, S_{\mu_1,l} + c_{r_1} \, S_{\mu_1,r} =  S_1^{Tot}.
\end{eqnarray}
Considering equation (\ref{Eq: general Lyapunov-2d}), we obtain the (1+1)-D space-time Lyapunov system as
\begin{eqnarray}
%\label{Eq: general Lyapunov-2d_22}
\nonumber
S_{\tau}\,\mathcal{U}\,M_{1}^T + M_{\tau}\,\mathcal{U}\, S_1^{{Tot}^{T}}+ \gamma M_{\tau}\,\mathcal{U}\,M_1^T =F.
\end{eqnarray}
We present a new class of basis and test functions yielding \textit{symmetric} stiffness matrices. Moreover, we compute exactly the corresponding mass matrices, which are either \textit{symmetric} and \textit{pentadiagonal}.
In the following, we extensively study the properties of the aforementioned matrices, allowing us to formulate a general fast linear solver for \eqref{Eq: total Lyapunov-2d1}.

%
%
%%%%%%%%%%%%%%%%%%%%%%%%%%%%%
\subsection{\textbf{Space of Basis Functions ($U_N$)}}
\label{Sec: Basis Func TSFA}
%%%%%%%%%%%%%%%%%%%%%%%%%%%%%
We construct the basis for the spatial discretization employing the Legendre polynomials defined as 
\begin{eqnarray}
\label{Eq: SpatialBasis PG TSFA}
\phi^{}_{m} (\, \xi \,) =
\sigma_m \,\big{(} P_{m+1}^{} (\xi)\, -\, P_{m-1}^{} (\xi)\big{)},\quad m=1,2,\cdots \quad and \, \, \, \xi \in [-1,1], 
\end{eqnarray}
%\\
where $\sigma_m = 2 + (-1)^{m}$. %and the $\mu$-dependent coefficient $\epsilon^{\mu}_{m_j} =(m-1-\mu)/(m-1)$.
The definition reflects the fact that for $\mu_j \leq 1/2$ and $1/2 \,\leq \nu_j \leq 1$, then both boundary conditions needs to be presented. Naturally, for the temporal basis functions only initial conditions are prescribed and the basis function for the temporal discretization is constructed based on the univariate poly-fractonomials \cite{zayernouri2013fractional} as
\begin{equation}
\label{Eq: TemporalBasis PG TSFA}
\psi^{\,\tau}_n(\eta) = {\sigma}_{n} (1+\eta)^{\tau}\,\, P_{n-1}^{-\tau\,,\, \tau} (\eta),\quad n=1,2,\cdots \quad and \, \, \, \eta\in [-1,1],
\end{equation}
for $n \geq 1$. With the notation established, we define the space-time trial space to be 
\begin{equation}
\label{Eq: Trial Space: PG}
U_N = 
%\begin{cases} 
span \Big\{    \Big( \psi^{\,\tau}_n \circ \eta \Big) ( t )
\prod_{j=1}^{d} \Big( \phi^{}_{m_j} \circ \xi_j\Big)  (x_j)\,
 : n = 1, \ldots, \mathcal{N}, \,m_j= 1, \ldots, \mathcal{M}_j\Big\},
\end{equation}
where $\eta(t) = 2t/T -1$ and $\xi_j(s) = 2\frac{s-a_j}{b_j-a_j} -1$. 

%
%%%%%%%%%%%%%%%%%%%%%%%%%%%%%
\subsection{\textbf{Space of Test Functions ($V_N$)}}
\label{Sec: Test Space PG TSFA}
%%%%%%%%%%%%%%%%%%%%%%%%%%%%%
%
We construct the \textit{spatial} test functions using  Legendre polynomial as well as the basis function in the Petrov-Galerkin method as 
\begin{eqnarray}
\label{Eq: SpatialTest PG TSFA}
\Phi^{}_{k} (\xi)
=
\widetilde{\sigma}_{k}\,\,\big{(} P_{k+1}^{} (\xi)\, -\, P_{k-1}^{} (\xi)\big{)},\quad k =1,2,\cdots \quad and \, \, \, \, \xi \in [-1,1],
\end{eqnarray}
where $ \widetilde{\sigma}_{k} = 2\,(-1)^{k} + 1$. Next, we define the \textit{temporal} test functions using the univariate poly-fractonomials
\begin{equation}
\label{Eq: TemporalTest PG TSFA}
\Psi^{\,\tau}_{r}(\eta) = \widetilde{\sigma}_{r} \,(1-\eta)^{\tau}\,\, P_{r-1}^{\tau\,,\, -\tau} (\eta),\quad r=1,2,\cdots \quad and \, \, \, \eta\in [-1,1],
\end{equation}
and we construct the corresponding space-time test space as
\begin{equation}
\label{Eq: Test Space: PG}
V_N = span \Big\{  \Big(\Psi^{\,\tau}_r \circ \eta\Big)(t)
\prod_{j=1}^{d} \Big( \Phi^{}_{k_j} \circ \xi_j\Big)(x_j)\,
 : r = 1, \ldots, \mathcal{N}, \,k_j= 1, \ldots, \mathcal{M}_j\Big\}.
\end{equation}
\vspace{0.01cm}
\begin{rem}
\label{Rem: Coeffs in Basis and Test Func TSFA}
The choices of $\sigma_{m}$ in \eqref{Eq: SpatialBasis PG TSFA} and \eqref{Eq: TemporalBasis PG TSFA}, also $\widetilde{\sigma}_{k}$ in \eqref{Eq: SpatialTest PG TSFA} and \eqref{Eq: TemporalTest PG TSFA}, result in the spatial/temporal mass and stiffness matrices being \textit{symmetric}, which are discussed in Theorems \ref{Thm: Temporal Stiffness},  \ref{Thm: Spatial Mass Matrices}, and \ref{Thm: Spatial Stiffness Matrix} in more details.
\end{rem}
\vspace{0.02cm}

%
%%%%%%%%%%%%%%%%%%%%%%%%%%%%%%%%%%%%%%
%
\subsection{\textbf{Implementation of PG Spectral Method}}
\label{Sec: Implementaiton of PG}
%
%%%%%%%%%%%%%%%%%%%%%%%%%%%%%%%%%%%%%%
%
We now seek the solution to \eqref{Eq: Genral_FPDE} in terms of a linear combination of elements in the space $U_N$ of the form
\begin{eqnarray}
\label{Eq: PG expansion}
&u_{N}(x,t) = 
\sum_{n=1}^\mathcal{N}
\,\,
\sum_{m_1=1}^{\mathcal{M}_1}
%\sum_{m_2= \lceil{\mu_2} \rceil}^{\mathcal{M}_2}
\cdots \sum_{m_d= 1}^{\mathcal{M}_d}
\hat u_{ n,m_1,\cdots,m_d} 
\Big[\psi^{\,\tau}_n(t)
\prod_{j=1}^{d} \phi^{}_{m_j}(x_j)
\Big]&
\end{eqnarray}
in $\Omega$. We enforce the corresponding residual 
\begin{eqnarray}
\label{Eq: Residual}
R_{N}(t,x_1,\cdots,x_d) &=& 
\prescript{}{0}{\mathcal{D}}_{t}^{2\tau} u_N   + \sum_{i=1}^{d} [c_{l_i}\prescript{}{a_i}{\mathcal{D}}_{x_i}^{2\mu_i} u_N +c_{r_i}\prescript{}{x_i}{\mathcal{D}}_{b_i}^{2\mu_i} u_N ]\,
\nonumber
\\
&-& \sum_{j=1}^{d} [\kappa_{l_j}\,\,\prescript{}{a_j}{\mathcal{D}}_{x_j}^{2\nu_j} u_N + \kappa_{r_j}\,\,\prescript{}{x_j}{\mathcal{D}}_{b_j}^{2\nu_j}u_N ]
+\gamma\, u_N -  f  \quad
\end{eqnarray}
to be $L^2$-orthogonal to $v_N \in V_N$, which leads to the finite-dimensional variational weak form in \eqref{Eq: Finite-dim PG method}. Specifically, by choosing $v_N = \Psi^{\,\tau}_r(t) \prod_{j=1}^{d} \Phi^{}_{k_j}(x_j)$, when $r = 1, \dots, \mathcal{N}$ and $k_j= 1, \dots, \mathcal{M}_j$, $j=1,2,\cdots,d$, we have 
\begin{eqnarray}
\label{Eq: general linear sys}
\sum_{n= 1}^\mathcal{N}
\sum_{m_1= 1}^{\mathcal{M}_1}
%\sum_{m_2= \lceil{\mu_2} \rceil}^{\mathcal{M}_2}
&\cdots &
\sum_{m_d= 1}^{\mathcal{M}_d}
\hat u_{ n,m_1,\cdots,m_d} 
\bigg( \{S_{\tau}\}_{r,n}\{M_1\}_{k_1,m_1} \cdots  \{M_d\}_{k_d,m_d}  
\nonumber
\\ 
\nonumber
&+&
 \sum_{i=1}^{d}
 [c_{l_i}
\{M_{\tau}\}_{r,n} \{M_1\}_{k_1,m_1} \cdots     
\{S_{\nu_i,l}\}_{k_i,m_i} 
\cdots 
\{M_d\}_{k_d,m_d}
\\ \nonumber
&+& c_{r_i}
\{M_{\tau}\}_{r,n} \{M_1\}_{k_1,m_1} \cdots     
\{S_{\nu_i,r}\}_{k_i,m_i} 
\cdots 
\{M_d\}_{k_d,m_d}] 
\\ \nonumber
&-& \sum_{j=1}^{d}
 \big{[}\kappa_{l_j}
\{M_{\tau}\}_{r,n}  \{ M_1\}_{k_1,m_1} \cdots     
\{S_{\nu_j,l}\}_{k_j,m_j}
\cdots 
\{M_d\}_{k_d,m_d}
\\ \nonumber
&+&\kappa_{r_j}
\{M_{\tau}\}_{r,n}  \{ M_1\}_{k_1,m_1} \cdots     
\{S_{\nu_j,r}\}_{k_j,m_j}
\cdots 
\{M_d\}_{k_d,m_d}\big{]} 
\\ \nonumber
&+&
\gamma \{M_{\tau}\}_{r,n} \{M_1\}_{k_1,m_1}  \cdots \{M_d\}_{k_d,m_d} \bigg)
\\ \nonumber
&=& F_{r,k_1,\cdots, k_d},
\end{eqnarray} 
where $S_{\tau}$ and $M_{\tau}$ denote, respectively, the temporal stiffness and mass matrices whose entries are defined as 
\begin{eqnarray}
\nonumber
\{S_{\tau}\}_{r,n}&=& 
\int_{0}^{T}
\prescript{}{0}{\mathcal{D}}_{t}^{\tau} \Big(\,\psi^{\tau}_n \circ \eta\Big)(t)
\prescript{}{t}{\mathcal{D}}_{T}^{\tau} \Big(\Psi^{\tau}_r \circ \eta \Big)(t)
\,
 dt,
\end{eqnarray}
and
\begin{eqnarray}
\nonumber
\{M_{\tau}\}_{r,n}= \int_{0}^{T} \,\Big(\Psi^{\tau}_r \circ \eta \Big)(t)
\Big(\,\psi^{\tau}_n \circ \eta\Big)(t)\,
 dt.
\end{eqnarray}
Moreover, $S_{\mu_j}$ and $M_{\mu_j}$, $j=1,2,\cdots,d$, are the corresponding spatial stiffness and mass matrices where the left-sided and right-sided entries of the spatial stiffness matrices are obtained as
\begin{eqnarray}
\nonumber
\{S_{\mu_j,l}\}_{k_j,m_j} =
\int_{a_j}^{b_j}
\prescript{}{a_j}{\mathcal{D}}_{x_j}^{\mu_j} \,
\Big( {\phi}^{}_{m_j} \circ \xi_j \Big)(x_j)\,
\prescript{}{x_j}{\mathcal{D}}_{b_j}^{\mu_j} 
\Big({\Phi}^{}_{k_j} \circ \xi_j\Big)(x_j)\,
dx_j = \{S_{\mu_j}\}_{k_j,m_j}, 
\\
\nonumber
\{S_{\mu_j,r}\}_{k_j,m_j} =
\int_{a_j}^{b_j}
\prescript{}{x_j}{\mathcal{D}}_{b_j}^{\mu_j} \,
\Big( {\phi}^{}_{m_j} \circ \xi_j \Big)(x_j)\,
\prescript{}{a_j}{\mathcal{D}}_{x_j}^{\mu_j} 
\Big({\Phi}^{}_{k_j} \circ \xi_j\Big)(x_j)\,
dx_j
= \{S_{\mu_j}\}_{k_j,m_j}^{T},
\end{eqnarray}
and the corresponding entries of the spatial mass matrix are given by
\begin{eqnarray}
\nonumber
\{M_j\}_{k_j,m_j} =
\int_{a_j}^{b_j}
\Big({\Phi}^{}_{k_j} \circ \xi_j\Big)(x_j)\,
\Big( {\phi}^{}_{m_j} \circ \xi_j \Big)(x_j)\,
dx_j. 
\end{eqnarray}
Moreover, the components of the load vector are computed as
\begin{eqnarray}
\label{Eq: general load matrix}
F_{r,k_1,\cdots, k_d} &=& \int_{\Omega}^{} f(t,x_1,\cdots,x_d) 
\Big(
\Psi^{\,\tau}_r \circ \eta \Big)(t)
\prod_{j=1}^{d} \Big(\Phi^{}_{k_j} \circ \xi_j\Big)(x_j)\, 
d\Omega. \quad 
\end{eqnarray} 
The linear system \eqref{Eq: general linear sys} can be exhibited as the following general Lyapunov equation 
\begin{eqnarray}
%\label{Eq: general Lyapunov}
\label{Eq: general Lyapunov - 2}
&\Big(&
S_{\tau} \otimes M_1 \otimes M_2 \cdots \otimes M_d  
\\ \nonumber
&+&\sum_{i=1}^{d}
 c_{l_i}
M_{\tau} \otimes M_1\otimes \cdots  \otimes S_{\mu_{i},l} \otimes M_{i+1}  \cdots \otimes M_d  \\ \nonumber
&+& \sum_{i=1}^{d} c_{r_i}
M_{\tau} \otimes M_1\otimes \cdots   \otimes S_{\mu_{i},r} \otimes M_{i+1}  \cdots \otimes M_d
\\ 
\nonumber
&-&
 \sum_{j=1}^{d}
\kappa_{l_j}
M_{\tau} \otimes M_1\otimes \cdots  \otimes S_{\nu_{j},l} \otimes M_{j+1}  \cdots \otimes M_d  \\ \nonumber
&-&  \sum_{j=1}^{d} \kappa_{r_j}
M_{\tau} \otimes M_1
\otimes \cdots   \otimes S_{\nu_{j},r} \otimes M_{j+1}  \cdots \otimes M_d 
\\ \nonumber
&+&
\gamma\,\, M_{\tau}\otimes M_1 \otimes M_2 \cdots \otimes M_d \Big) \, \mathcal{U}= F.
\end{eqnarray} 
Let
\begin{equation}
c_{l_i} \times S_{\mu_{i},l} + c_{r_i} \times S_{\mu_{i},r}-\kappa_{l_i} \times S_{\nu_{i},l} -\kappa_{r_i} \times S_{\nu_{i},r}= S^{\, {Tot}}_{i}.
\end{equation} 
Considering the fact that all the aforementioned stiffness and mass matrices are \textit{symmetric}, $ S_{\mu_{i},l}$, $ S_{\mu_{i},r}$, $ S_{\nu_{i},l}$, and $ S_{\nu_{i},r}$ can be replaced by $S^{\, {Tot}}$ which remains symmetric. Therefore,
\begin{eqnarray}
\label{Eq: general Lyapunov}
&&\Big(
\,\,S_{\tau} \otimes M_1 \otimes M_2 \cdots \otimes M_d 
\\ \nonumber
&&+
 \sum_{i=1}^{d} 
 [\,\,
M_{\tau} \otimes M_1\otimes \cdots   \otimes M_{i-1} \otimes S_{i}^{\, {Tot}} \otimes M_{i+1}  \cdots \otimes M_d \,\,]
\\ \nonumber
&&+\gamma\,\, M_{\tau}\otimes M_1 \otimes M_2 \cdots \otimes M_d \Big) \, \mathcal{U}= F,
\end{eqnarray} 
in which $\otimes$ represents the Kronecker product, $F$ denotes the multi-dimensional load matrix whose entries are given in \eqref{Eq: general load matrix}, and $\mathcal{U}$ denotes the corresponding multi-dimensional matrix of unknown coefficients with entries $\hat u_{ n,m_1,\cdots,m_d} $. 

In the Theorems \ref{Thm: Temporal Stiffness}, \ref{Thm: Spatial Mass Matrices}, and \ref{Thm: Spatial Stiffness Matrix}, we study the properties of the aforementioned matrices. Besides, we present efficient ways of deriving the spatial mass and the temporal stiffness matrices analytically and exact computation of the temporal mass and the spatial stiffness matrices through proper quadrature rules. 
%Moreover, we provided the discrete stability analysis high-dimensional advection-dispersion problem in details in (PART-II). We also obtained [PART-II] the convergence rate of the scheme while high-order dreivatives are assumed to be finite.

\begin{thm}
\label{Thm: Temporal Stiffness}
The temporal stiffness matrix $S_{\tau}$ corresponding to the time-fractional order $\tau \in (0,1)$ is a diagonal $\mathcal{N}\times \mathcal{N}$ matrix, whose entries are obtained as
\begin{eqnarray}
\nonumber
\{S_{\tau}\}_{r,n}= \widetilde{\sigma}_r\,\sigma_n \frac{\Gamma(n+\tau)}{\Gamma(n)} \, \frac{\Gamma(r+\tau)}{\Gamma(r)}
\Big(\frac{2}{T} \Big)^{2\tau-1}\, \frac{2}{2n-1} \, \delta_{r,n}, \quad r,n=1,2,\cdots, \,\mathcal{N}.
\end{eqnarray}
Moreover, the entries of temporal mass matrices $M_{\tau}$ can be computed exactly by employing a Gauss-Lobatto-Jacobi (GLJ) rule with respect to the weight function $(1-\eta)^{\tau}(1+\eta)^{\tau}$, $\eta \in [-1,1]$, where $\alpha=\tau/2$. Moreover, $M_{\tau}$ is symmetric.
\begin{proof}
See \cite{zayernouri2015unified}.
\end{proof}
\end{thm}

\begin{thm}
\label{Thm: Spatial Mass Matrices}
The spatial mass matrix $M$ is a penta-diagonal $\mathcal{M}\times \mathcal{M}$ matrix, whose entries are explicitly given as
\begin{eqnarray}
\quad {M}_{k,r}&=&{\widetilde{\sigma}}_k\,{\sigma}_r\, \Big{[}\frac{2}{2k+3}\delta_{k,r}-\frac{2}{2k+3}\delta_{k+1,r-1}
-\frac{2}{2k-3}\delta_{k-1,r+1}+\frac{2}{2k-3}\delta_{k-1,r-1} \Big{]}. \quad \quad
\end{eqnarray}

\begin{proof}
The (k, r)\textit{th}-entry of the spatial mass matrix is given by
\begin{eqnarray}
\label{Eq: Spatial Mass_1}
&{M}_{k,r}= 
\int_{a}^{b}
\,\phi^{}_r \circ \xi ( x ) \, \Phi^{}_k \circ \xi ( x ) 
\,
 d{x} = \big{(} \frac{b-a}{2}  \big{)} \, \int_{-1}^{1}
\,\phi^{}_r  ( \xi ) \Phi^{}_k ( \xi ) 
\,
 d{\xi} ,&
\end{eqnarray}
where $\xi = 2 \frac{x-a}{b-a}-1$ and $\xi \in (-1, \, 1)$. Substituting the spatial basis/test functions, we have 
\begin{eqnarray}
%\label{Eq: total Lyapunov-2d}
\label{Eq: total Lyapunov-2d-2}
&{M}_{k,r}= 
\big{(} \frac{b-a}{2}  \big{)} \,{\widetilde{\sigma}}_k\,{\sigma}_r\,\big[\,\widetilde{M}_{k,r}^{}-\widetilde{M}_{k+1,r-1}^{}-\widetilde{M}_{k-1,r+1}^{}+\widetilde{M}_{k,r}^{}  \big], &
\end{eqnarray}
in which  
\begin{eqnarray}
\widetilde{M}_{i,j}^{} = \,\int_{-1}^{1}
\, P_{i} (\xi) \, P_{j} (\xi)  \,
 d{\xi}&& = \frac{2}{2i+1}\delta_{ij}.
\end{eqnarray}
Therefore, we have
\begin{eqnarray}
\nonumber
{M}_{k,r}&=& \big{(} \frac{b-a}{2}  \big{)} \,{\widetilde{\sigma}}_k\,{\sigma}_r\, \Big[\frac{2}{2k+3}\delta_{k,r}-\frac{2}{2k+3}\delta_{k+1,r-1}
-\frac{2}{2k-3}\delta_{k-1,r+1}+\frac{2}{2k-3}\delta_{k,r} \Big]
\end{eqnarray}
as a pentadiagonal matrix. Moreover,
\begin{eqnarray}
{M}_{r,k}&=& \big{(} \frac{b-a}{2}  \big{)} \,{\widetilde{\sigma}}_r\,{\sigma}_k\, \Big[\frac{2}{2r+3}\delta_{r,k}-\frac{2}{2r+3}\delta_{r+1,k-1}
-\frac{2}{2r-3}\delta_{r-1,k+1}+\frac{2}{2r-3}\delta_{r,k} \Big] \nonumber
\\
\nonumber
&=& {M}_{k,r}.
\end{eqnarray}
\end{proof}
\end{thm}

%%%%%%%%%%%%%%%%%%%%%%%%%%%%%%%%%%%%%%%%%%%
%%%%%%%%%%%%%%%%%%%%%%%%%%%%%%%%%%%%%%%%%%%%
%
\begin{thm}
\label{Thm: Spatial Stiffness Matrix}
The total spatial stiffness matrix $S^{Tot}_{i}$ is symmetric and its entries can be exactly computed as:
%
%To figure out the total spatial stiffness matrix, $S_{\mu_1}$ and $S_{\nu_1}$ are evaluated here. As we know from sec 3,
\begin{eqnarray}
\label{Eq: total Lyapunov-2d}
c_{l_i} \times S_{\mu_{i},l} + c_{r_i} \times S_{\mu_{i},r}-\kappa_{l_i} \times S_{\nu_{i},l}-\kappa_{r_i} \times S_{\nu_{i},r}= S^{\, {Tot}}_{i}.
\end{eqnarray}
where $i=1,2,\cdots,d$.

\begin{proof}
Regarding the definition of stiffness matrix, we have
\begin{eqnarray}
 \{ S_{\mu_i,l}\}_{r,n}
&=& 
\int_{a_i}^{b_i}
\prescript{}{a_i}{\mathcal{D}}_{x_i}^{\mu_i} \Big(\,\phi^{}_n  ( x_i ) \Big)
\prescript{}{x_i}{\mathcal{D}}_{b_i}^{\mu_i} \Big(\Phi^{}_r ( x_i ) \Big)
\,
 d{x_i},
\nonumber
\\
&= &
\big{(} \frac{b_i-a_i}{2}  \big{)} ^{-2\mu_i+1} \,{\widetilde{\sigma}}_r\,\sigma_n\, \int_{-1}^{1} \prescript{}{-1}{\mathcal{D}}_{\xi_i}^{\mu_i} \Big(\, P_{n+1} - P_{n-1} \Big)
\prescript{}{{\xi}_i}{\mathcal{D}}_{1}^{\mu_i} \Big(\, P_{k+1} - P_{k-1} \Big)
\,
 d{\xi_i}
\nonumber
\\
&=&
\big{(} \frac{b_i-a_i}{2}  \big{)} ^{-2\mu_i+1} \,{\widetilde{\sigma}}_r\,\sigma_n\, \Big[\,\widetilde{S}_{r+1,n+1}^{\, \mu_i}-\widetilde{S}_{r+1,n-1}^{\, \mu_i}-\widetilde{S}_{r-1,n+1}^{\, \mu_i}+\widetilde{S}_{r-1,n-1}^{\, \mu_i}  \Big],   \quad \quad
\end{eqnarray}
where
\begin{eqnarray}
&\widetilde{S}_{r,n}^{\, \mu_i} = \,\int_{-1}^{1}
\prescript{}{-1}{\mathcal{D}}_{\xi_i}^{\mu_i} \Big(\, P_{n} (\xi_i) \Big)
\prescript{}{\xi_i}{\mathcal{D}}_{1}^{\mu_i} \Big(\, P_{r} (\xi_i) \Big)
\,
 d{\xi_i}&
 \nonumber 
 \\
&=
\int_{-1}^{1}\, \frac{\Gamma(r+1)}{\Gamma(r-\mu_i+1)}\,  \frac{\Gamma(n+1)}{\Gamma(n-\mu_i+1)} \, {(1+\xi_i)}^{-\mu_i} {(1-\xi_i)}^{-\mu_i} \, P^{ -\mu_i,\mu_i}_{r} \, P^{ \mu_i,-\mu_i}_{n} d{\xi_i}.&
\nonumber
\end{eqnarray}
$\widetilde{S}_{r,n}^{\, \mu_i}$ can be computed accurately using Guass-Jacobi quadrature rule as
\begin{eqnarray}
 \widetilde{S}_{r,n}^{\, \mu_i}
&=&
\,\, \frac{\Gamma(r+1)}{\Gamma(r-\mu_i+1)}\,  \frac{\Gamma(n+1)}{\Gamma(n-\mu_i+1)}
 \sum_{q=1}^{Q} w_q \,\,
P_{r}^{-\mu_i,\, \mu_i} (\xi_q)
P_{n}^{\mu_i,\, -\mu_i} (\xi_q), 
\end{eqnarray}
in which $\mathcal{Q} \ge \mathcal{N} +2 $ represents the minimum number of GJ quadrature points $\{\xi_q\}_{q=1}^{\mathcal{Q}}$, associated with the weigh function $(1-\xi)^{-\mu_i} (1+\xi)^{-\mu_i}$, for \textit{exact} quadrature, and $\{w_q\}_{q=1}^{Q}$ are the corresponding quadrature weights. 
Employing the property of the Jacobi polynomials where $P^{\alpha, \beta}_n(-x_i) = (-1)^n P^{ \beta,\alpha}_n(x_i)$, we can re-express $\widetilde{S}_{r,n}^{\, \mu_i}$ as $(-1)^{(r+n)}\,\widetilde{S}_{n,r}^{\, \mu_i}$. Accordingly,
\begin{eqnarray}
\label{Eq: Spatial stiffness_2}
\{S_{ \mu_i}\}_{r,n}
\nonumber
&=&
\big{(} \frac{b_i-a_i}{2}  \big{)} ^{-2\mu_i+1} \,{\widetilde{\sigma}}_r\,\sigma_n\, \Big[(-1)^{(n+r+2)}\,\widetilde{S}_{n+1,r+1}^{\, \mu_i}-(-1)^{(n+r)}\,\widetilde{S}_{n+1,r-1}^{\, \mu_i}
\\
\nonumber
&&-(-1)^{(n+r)}\,\widetilde{S}_{n-1,r+1}^{\, \mu_i}+(-1)^{(n+r-2)}\,\widetilde{S}_{n-1,r-1}^{\, \mu_i}  \Big]
\nonumber
\\
&=&
{\widetilde{\sigma}}_r\,\sigma_n\, (-1)^{(n+r)}\,\Big[\widetilde{S}_{n+1,r+1}^{\, \mu_i}-\widetilde{S}_{n+1,r-1}^{\, \mu_i}-\widetilde{S}_{n-1,r+1}^{\, \mu_i}+\widetilde{S}_{n-1,r-1}^{\, \mu_i}  \Big].
\end{eqnarray} 
According to \eqref{Eq: Spatial stiffness_2}, 
\begin{eqnarray}
\label{Eq: Spatial stiffness_3}
\{S_{\, \mu_i}\}_{r,n} =  \{S_{\, \mu_i}\}_{n,r}\, \times \, \frac{{\widetilde{\sigma}}_r\,\sigma_n\,}{{\widetilde{\sigma}}_n\,\sigma_r\,} (-1)^{(n+r)}.
\end{eqnarray}
In fact, ${\widetilde{\sigma}}_r$ and $\sigma_n$ are chosen such that $(-1)^{(n+r)}$ is canceled. Furthermore,
\begin{eqnarray}
\label{eq:symm_111}
\{ S_{\mu_i,r}\}_{r,n}
&=& 
\int_{a_i}^{b_i}
\prescript{}{a_i}{\mathcal{D}}_{x_i}^{\mu_i} \Big(\,\Phi^{}_r ( x_i ) \Big)
\prescript{}{x_i}{\mathcal{D}}_{b_i}^{\mu_i} \Big(\phi^{}_n ( x_i ) \Big)
\,
d{x_i},
\nonumber
\\
&=&\int_{a_i}^{b_i}
\prescript{}{a_i}{\mathcal{D}}_{x_i}^{\mu_i} \Big(\,\phi^{}_n  ( x_i ) \Big)
\prescript{}{x_i}{\mathcal{D}}_{b_i}^{\mu_i} \Big(\Phi^{}_r ( x_i ) \Big)
\,
d{x_i},
\nonumber
\\
&= &
\{ S_{\mu_i,l}\}_{n,r}, 
\end{eqnarray}
where $\{ S_{\mu_i,l}\}_{n,r}=\{ S_{\mu_i,l}\}_{r,n}=\{ S_{\mu_i,r}\}_{r,n}=\{ S_{\mu_i}\}_{r,n}$ due to symmetry of $ S_{\mu_i,l}$ and $ S_{\mu_i,r}$. Similar to \eqref{eq:symm_111}, we get $\{ S_{\nu_i,l}\}_{r,n}=\{ S_{\nu_i,r}\}_{r,n}=\{ S_{\nu_i}\}_{r,n} $; therefore,
\begin{eqnarray}
&-(\kappa_{l_i}+\kappa_{r_i}) \,S_{\nu_i} + (c_{l_i}+c_{r_i}) \, S_{\mu_i} =  S_i^{Tot}.
\end{eqnarray}
 Hence it can be easily concluded that the stiffness matrix $S^{\, \mu_i}_{n,r}$, $S^{\, \nu_i}_{n,r}$ and thereby $\{S^{Tot}_{i}\}_{n,r}$ as the sum of two symmetric matrices are symmetric.
\end{proof}
\end{thm}

%
%%%%%%%%%%%%%%%%%%%%%%%%%%%%%%%%%%%%%%
%
\section{\textbf{Unified Fast FPDE Solver}}
\label{Sec: FastSolver FPDE}
%
%%%%%%%%%%%%%%%%%%%%%%%%%%%%%%%%%%%%%%
%
We formulate a closed-form solution for the Lyapunov system \eqref{Eq: general Lyapunov} in terms of the generalised eigensolutions that can be computed very efficiently, leading to the following unified fast solver for the development of Petrov-Galerkin spectral method.

\begin{thm}
\label{Thm: fast solver}
Let $\{ {\vec{e}_{}}^{\, j}    ,    \lambda^{\, j}\,  \}_{m_j=1}^{\mathcal{M}_j}$ be the set of general eigen-solutions of the spatial stiffness matrix $S^{Tot}_j$ with respect to the mass matrix $M_{j}$. Moreover, let $\{ {\vec{e}_{}}^{\,\,\tau}    ,    \lambda^{\tau}_{}\,  \}_{n=1}^{\mathcal{N}}$ be the set of general eigen-solutions of the temporal mass matrix $M_{\tau}$ with respect to the stiffness matrix $S_{\tau}$. 

(I) if $d>1$, then the multi-dimensional matrix of unknown coefficients $\mathcal{U}$ is explicitly obtained as
\begin{equation}
\label{Eq: thm u expression in terms of k}
\mathcal{U} = 
\sum_{n=1}^{\mathcal{N}}
\,\,
\sum_{m_1= 1}^{\mathcal{M}_1}
\cdots 
\sum_{m_d= 1}^{\mathcal{M}_d}
\kappa_{ n,m_1,\cdots,\,m_d  } \,
\,\vec{e}_n^{\,\,\tau}\,
\otimes
\,{\vec{e}_{m_1}}^{\, 1}\,\,
\otimes
\cdots
\otimes
\,{\vec{e}_{m_d}}^{\, d},
\end{equation}
where $\kappa_{ n,m_1,\cdots,\,m_d }$ are given by 
\begin{eqnarray}
\label{Eq: thm k fraction_1}
\kappa_{ n,m_1,\cdots,\,m_d  } =  \frac{(\,\vec{e}_n^{\,\,\tau}
\,{\vec{e}_{m_1}}^{\, 1}
\cdots
\,{\vec{e}_{m_d}}^{\, d}) F}
{
\Big[
(\vec{e}_n^{\,\,\tau^T} S_{\tau} \vec{e}_n^{\,\,\tau})
\prod_{j=1}^{d} ((\vec{e}_{m_j}^{j})^{T}   M_{j} {\vec{e}_{m_j}}^{})
\Big]
\Lambda_{n,m_1,\cdots,m_d}
},
\end{eqnarray}
in which the numerator represents the standard multi-dimensional inner product, and $\Lambda_{n,m_1,\cdots,m_d}$ are obtained in terms of the eigenvalues of all mass matrices as
\begin{eqnarray}
  \nonumber
&\Lambda_{n,m_1,\cdots, m_d} = \Big[
(1+\gamma\,\, 
\lambda^{\tau}_n)
+
\lambda^{\tau}_n
 \sum_{j=1}^{d}
(
\lambda^{j}_{m_j}
)
\Big].  &
\end{eqnarray}
(II) If $d=1$, then the two-dimensional matrix of the unknown solution $\mathcal{U}$ is obtained as
\begin{equation}
\nonumber
\mathcal{U} = 
\sum_{n=1}^{\mathcal{N}}
\,\,
\sum_{m_1= 1}^{\mathcal{M}_1}
\kappa_{ n,m_1 } \,
\,\vec{e}_n^{\,\,\tau}\,
\,({\vec{e}_{m_1}}^{\, 1})^{T},
\end{equation}
where $\kappa_{n,m_1}$ is explicitly obtained as 
\begin{eqnarray}
\nonumber
\kappa_{ n,m_1  } = 
 \frac{
\vec{e}_n^{\,\,\tau^T}
F
\,{\vec{e}_{m_1}}^{\, 1}
}
{
(\vec{e}_n^{\,\,\tau^T} S_{\tau} \vec{e}_n^{\,\,\tau})
(({\vec{e}_{m_1}}^{\, 1})^{T}   M_{1} {\vec{e}_{m_1}}^{\, 1})
\Big[
(1+\gamma\,\, 
\lambda^{\tau}_n)
+\,
\lambda^{\tau}_n
\,
\lambda^{\, 1}_{m_1}
\Big]
}.
\end{eqnarray}
\end{thm}

\begin{proof}
%PPPPPPPPPPPPPPPPPPPPPPPPPP
%
Let us consider the following generalised eigenvalue problems as
\begin{eqnarray}
\label{Eq: Spatial EigenProblem}
&S^{\, Tot}_{j}\,{\vec{e}_{m_j}}^{\, j} =\lambda^{\, j}_{m_j}\,M_{j}\,{\vec{e}_{m_j}}^{\, j},\quad m_j=1,\cdots, \mathcal{M}_j, \quad j=1,2,\cdots, d,& 
\end{eqnarray}
and
\begin{eqnarray}
\label{Eq: Temporal EigenProblem}
&M_{\tau}\,{\vec{e}_n}^{\,\,\tau}=\lambda^{\tau}_n\,S_{\tau}\,{\vec{e}_n}^{\,\,\tau},\quad n=1,2,\cdots, \mathcal{N}.&
\end{eqnarray}
Having the spatial and temporal eigenvectors determined in equations \eqref{Eq: Temporal EigenProblem} and \eqref{Eq: Spatial EigenProblem}, we can represent the unknown coefficient matrix  $\mathcal{U} $  in  \eqref{Eq: PG expansion} in terms of the aforementioned eigenvectors as
\begin{equation}
\label{Eq: rep U}
\mathcal{U} = 
\sum_{n=1}^{\mathcal{N}}
\,\,
\sum_{m_1= 1}^{\mathcal{M}_1}
\cdots 
\sum_{m_d= 1}^{\mathcal{M}_d}
\kappa_{ n,m_1,\cdots,\,m_d  } \,
\,\vec{e}_n^{\,\,\tau}\,
\otimes
\,{\vec{e}_{m_1}}^{\, 1}\,\,
\otimes
\cdots
\otimes
\,{\vec{e}_{m_d}}^{\, d},
\end{equation}
where $\kappa_{ n,m_1,\cdots,\,m_d } $ 
are obtained as follows. First, we take the multi-dimensional inner product of $\,\vec{e}_{q}^{\,\,\tau}
\,{\vec{e}_{p_1}}^{\, 1}
\cdots
\,{\vec{e}_{p_d}}^{\, d}$ on both sides of the Lyapunov equation \eqref{Eq: general Lyapunov} as
\begin{eqnarray}
\label{Eq: Lyap_2}
\nonumber
&(\,\vec{e}_q^{\,\,\tau}
\,{\vec{e}_{p_1}}^{\, 1}
\,{\vec{e}_{p_2}}^{\, 2}
\cdots
\,{\vec{e}_{p_d}}^{\, d}) 
\Big[
\,\,
S_{\tau} \otimes M_{1} \otimes \cdots \otimes M_{d}  & 
\\ \nonumber
& +
 \sum_{j=1}^{d}
\, 
 [
M_{\tau} \otimes M_{1}\otimes \cdots   \otimes M_{{j-1}} \otimes S_{{j}}^{Tot} \otimes M_{{j+1}}  \cdots \otimes M_{d} ] &
\\ \nonumber
&+\gamma\,\, M_{\tau} \otimes M_{1} \otimes \cdots \otimes M_{d} \Big] \mathcal{U}
= 
(\,\vec{e}_q^{\,\,\tau}
\,{\vec{e}_{p_1}}^{\, 1}
\cdots
\,{\vec{e}_{p_d}}^{\, d}) F. &
\hspace{1cm} 
\end{eqnarray} 
Then, by replacing (\ref{Eq: Spatial EigenProblem}) and (\ref{Eq: Temporal EigenProblem}) into (\ref{Eq: thm k fraction_1}) and re-arranging the terms, we get
\begin{eqnarray}
\label{Eq: fast solver_1112}
\nonumber
\sum_{n=1}^{\mathcal{N}}
\,\,
\sum_{m_1= 1}^{\mathcal{M}_1}
&\cdots & 
\sum_{m_d= d}^{\mathcal{M}_d}
\kappa_{ n,m_1,\cdots,\,m_d } \times\, 
\Big(
\,\,
\vec{e}_q^{\,\,\tau^T} S_{\tau} \vec{e}_n^{\,\,\tau}
\,\,\,
(\vec{e}_{p_1}^{\, 1})^{T}   M_{1} \vec{e}_{m_1}^{\, 1}
\,\, \,
\cdots
\,\,\,
(\vec{e}_{p_d}^{\, d})^{T}    M_{d} \, \vec{e}_{m_d}^{\, d}
\\ \nonumber
&+&
 \sum_{j=1}^{d}
\,
\vec{e}_q^{\,\,\tau^T} M_{\tau} \vec{e}_n^{\,\,\tau}
\,
(\vec{e}_{p_1}^{\, 1})^{T}   M_{1} \vec{e}_{m_1}^{\, 1}
\,
\cdots
\,
(\vec{e}_{p_j}^{\, j})^{T}   S_{{j}}^{\, {Tot}} \vec{e}_{m_{j}}^{\, j}
\,
(\vec{e}_{p_{j+1}}^{\, j+1})^{T}   M_{{j+1}} \vec{e}_{m_{j+1}}^{\, j+1}
\,
\cdots (\vec{e}_{p_d}^{\, d})^{T}   M_{d} \vec{e}_{m_d}^{\, d}
\\ \nonumber
&+&\gamma\,\, 
\vec{e}_q^{\,\,\tau^T} M_{\tau} \vec{e}_n^{\,\,\tau}
\,
(\vec{e}_{p_1}^{\, 1})^{T}   M_{1} \vec{e}_{m_1}^{\, 1}
\,
(\vec{e}_{p_1}^{\, 2})^{T}   M_{2} \vec{e}_{m_2}^{\, 2}
\,
\cdots
\,
(\vec{e}_{p_d}^{\, d})^{T}   M_{d} \vec{e}_{m_d}^{\, d}
\Big)
\\ \nonumber
&= &
(\,\vec{e}_q^{\,\,\tau}
\,{\vec{e}_{p_1}}^{\, 1}
\,{\vec{e}_{p_2}}^{\, 2}
\cdots
\,{\vec{e}_{p_d}}^{\, d}) F.
\end{eqnarray} 
Recalling that $S_{j}^{\, {Tot}} {\vec{e}_{m_j}}^{\, j}\,   = (\lambda_{m_j}^{\, j} M_{\, j} \,{\vec{e}_{m_j}}^{\, j})$ and $M_{\tau}\,{\vec{e}_n}^{\,\,\tau} =(\lambda^{\tau}_n\,S_{\tau}\,{\vec{e}_n}^{\,\,\tau})$, we have
\begin{eqnarray}
\nonumber
\sum_{n=1}^{\mathcal{N}}
\,\,
\sum_{m_1= 1}^{\mathcal{M}_1}
&\cdots & 
\sum_{m_d= 1}^{\mathcal{M}_d}
\kappa_{ n,m_1,\cdots,\,m_d  } \, \Big(
\,
\vec{e}_q^{\,\,\tau^T} S_{\tau} \vec{e}_n^{\,\,\tau}
\,
(\vec{e}_{p_1}^{\, 1})^{T}   M_{1} {\vec{e}_{m_1}}^{\, 1}
\, 
(\vec{e}_{p_2}^{\, 2})^{T}   M_{2} {\vec{e}_{m_2}}^{\, 2}
\,
\cdots (\vec{e}_{p_d}^{\, d})^{T} M_{d}\, {\vec{e}_{m_d}}^{\, d}) 
\\ \nonumber
& +&
 \sum_{j=1}^{d}
\,
\vec{e}_q^{\,\,\tau^T} (\lambda^{\tau}_n\,S_{\tau}\,{\vec{e}_n}^{\,\,\tau})
\,
(\vec{e}_{p_1}^{\, 1})^{T} M_{1} {\vec{e}_{m_1}}^{\, 1}
\cdots 
(\vec{e}_{p_j}^{\, j})^{T}  (\lambda^{\, j}_{m_{j}} M_{{j}} 
\vec{e}_{m_{j}}^{\, j}) 
\cdots
(\vec{e}_{p_d}^{\, d})^{T}   M_{d}\, {\vec{e}_{m_d}}^{\, d}
\,
\\ \nonumber
& +& \gamma\,\, 
\vec{e}_q^{\,\,\tau^T} (\lambda^{\tau}_n\,S_{\tau}\,{\vec{e}_n}^{\,\,\tau})
\,
(\vec{e}_{p_1}^{\, 1})^{T}   M_{1} {\vec{e}_{m_1}}^{\, 1}
\,
(\vec{e}_{p_2}^{\, 2})^{T}   M_{2} {\vec{e}_{m_2}}^{\, 2}
\,
%\\ \nonumber
%&\,&
\cdots
\,
(\vec{e}_{p_d}^{\, d})^{T}   M_{d} {\vec{e}_{m_d}}^{\, d}
\,
\Big)
\\
\nonumber
&
=& (\,\vec{e}_q^{\,\,\tau}
\,{\vec{e}_{p_1}}^{\, 1}
\,{\vec{e}_{p_2}}^{\, 2}
\cdots
\,{\vec{e}_{p_d}}^{\, d}) F.
\end{eqnarray} 
Therefore,
\begin{eqnarray}
\label{Eq: thm k fraction}
\nonumber
\kappa_{ n,m_1,\cdots,\,m_d  } =  \frac{(\,\vec{e}_n^{\,\,\tau}
\,{\vec{e}_{m_1}}^{\, 1}
\cdots
\,{\vec{e}_{m_d}}^{\, d}) F}
{
\Big[
(\vec{e}_n^{\,\,\tau^T} S_{\tau} \vec{e}_n^{\,\,\tau})
\prod_{j=1}^{d} ((\vec{e}_{m_j}^{\, j})^{T}   M_{j} {\vec{e}_{m_j}}^{\, j})
\Big] \times
\Big[
(1+\gamma\,\, 
\lambda^{\tau}_n)
+
\lambda^{\tau}_n
 \sum_{j=1}^{d}
(
\lambda^{\, j}_{m_j}
)
\Big]
}.
\end{eqnarray}
Then, we have
\begin{eqnarray}
\nonumber
&\sum_{n=1}^{\mathcal{N}}
\,\,
\sum_{m_1= 1}^{\mathcal{M}_1}
%\sum_{m_2= \lceil{\mu_2} \rceil}^{\mathcal{M}_2}
\cdots 
\sum_{m_d= 1}^{\mathcal{M}_d}
\kappa_{ n,m_1,\cdots,\,m_d  } \,
(\vec{e}_q^{\,\,\tau^T} S_{\tau} \vec{e}_n^{\,\,\tau})
((\vec{e}_{p_1}^{\, 1})^{T}  M_{1} {\vec{e}_{m_1}}^{\, 1})
\cdots
((\vec{e}_{p_d}^{\, d})^{T} M_{d} {\vec{e}_{m_d}}^{\, d})&
\\ \nonumber
&\times
\Big[
(1+\gamma\,\, 
\lambda^{\tau}_n)
+
\lambda^{\tau}_n
 \sum_{j=1}^{d}
(
\lambda^{\, j}_{m_j}
)
\Big]=
(\,\vec{e}_q^{\,\,\tau}
\,{\vec{e}_{p_1}}^{\, 1}
\,{\vec{e}_{p_2}}^{\, 2}
\cdots
\,{\vec{e}_{p_d}}^{\, d}) F.&
\end{eqnarray} 
Due to the fact that the spatial Mass ${M_{\, j}}$ and temporal stiffness matrices ${S_{\, \tau}}$  are diagonal (see Theorems \ref{Thm: Spatial Mass Matrices} and \ref{Thm: Temporal Stiffness}), we have $(\vec{e}_q^{\,\,\tau^T} S_{\tau} \vec{e}_n^{\,\,\tau}) = 0$ if $q\ne n$,  and also $((\vec{e}_{p_j}^{\, j})^{T}   M_{j} {\vec{e}_{m_j}}^{\, j}) = 0$ if $p_j \ne m_j$, which completes the proof for the case $d>1$.

Following similar steps for the two-dimensional problem, it is easy to see that if $d=1$, the relationship for $\kappa$ can be derived as 
\begin{eqnarray}
\label{Eq: Fast FPDE Solver d 1111}
\kappa_{ q,p_1 } = 
 \frac{
\vec{e}_q^{\,\,\tau^T}
F
\,{\vec{e}_{p_1}}^{\, 1}
}
{
(\vec{e}_q^{\,\,\tau^T} S_{\tau} \vec{e}_q^{\,\,\tau})
((\vec{e}_{p_1}^{\, 1})^{T}   M_{1} {\vec{e}_{p_1}}^{\, 1})
\Big[
(1+\gamma\,\, 
\lambda^{\tau}_n)
+\,
\lambda^{\tau}_n
\,
\lambda^{\, 1}_{m_1}
\Big]
}.
\end{eqnarray}
%
%
%PPPPPPPPPPPPPPPPPPPPPPPPPP
In \ref{Sec: Computational Considerations}, we present a computational method for the fast solver which reduces the computational cost significantly.
\end{proof}

\subsection{\textbf{Computational Considerations}}
\label{Sec: Computational Considerations}
%
%%%%%%%%%%%%%%%%%%%%%%%%%%%%%%%%%%%%%%
%
Employing the fast solver in $(1+d)$ dimensional problem $d \, \geq \, 1$ reduces the dominant computational cost of the eigensolver from $\mathcal{O}(N^{2(1+d)})$ to  $\mathcal{O}(N^{2+d})$, which becomes even more efficient in higher dimensional problems. This approach is extensively discussed in \cite{zayernouri2015unified}.

\section{\textbf{Numerical Tests}}
\label{num-test}

We now examine the unified PG spectral method and the corresponding unified fast solver (\ref{Eq: rep U}) and (\ref{Eq: Fast FPDE Solver d 1111}) for (\ref{Eq: Genral_FPDE}) in the context of several numerical test cases in order to investigate the spectral/exponential rate of convergence in addition to the computational efficiency of the scheme. The corresponding force term $f$ in (\ref{Eq: Genral_FPDE}) is obtained in Appendix for the following test cases, listed as:
\vspace{0.3 in}

\noindent \textbf{Test case (I):} (smooth solutions with finite regularity) we consider the following exact solution to perform the temporal \textit{$p$-refinement} as
\begin{equation}
\label{test11111}
u^{exact} = t^{p_1} \times \Big{(}(1+x)^{p_2} - \epsilon (1+x)^{p_3}\Big{)}, 
\end{equation}
where $p_1 = 7\frac{2}{3}$, $p_2 = 6\frac{1}{3}$, $p_3 = 6\frac{2}{7}$ and $t \in [0, \, 2]$ and $ x \in [-1, \, 1].$
\vspace{0.3 in}

\noindent \textbf{Test case (II):} (spatially smooth function) we consider 
\begin{equation}
u^{exact} = t^{p_1} \times sin[n\pi \, (1+x)], 
\end{equation}
where $n=1$ and $p_1 = 6\frac{1}{3}$, for the exponential \textit{$p$-refinement}.
\vspace{0.3 in}

\noindent \textbf{Test case (III):} (high-dimensional problems) to perform the \textit{$p$-refinement} in higher dimensions ($d = 2, \, 3$), we choose the exact solution
\begin{equation}
\label{testIII}
u^{exact} = t^{p_1} \times \prod_{i=1}^{d}\, \Big{(}(1+x_i)^{p_{2i}} - \epsilon (1+x)^{p_{2i+1}}\Big{)},
\end{equation}
where $p_1 = 7\frac{2}{3}$, $p_2 = 6\frac{1}{3}$, $p_3 = 6\frac{2}{7}$, $p_4=7\frac{4}{5}$, $p_5=7\frac{1}{7}$, $p_6=7\frac{3}{5}$, $p_7=7\frac{1}{7}$  and $\epsilon_1 = 2^{p_2-p_3}$, $\epsilon_2 = 2^{p_4-p_5}$, $\epsilon_3 = 2^{p_6-p_7}$ in the hypercube domain as $\underbrace{[0, \, 1]\times[-1, \, 1]\times \cdots\times[-1, \, 1]}_{d\, \, times}$.
\vspace{0.3 in}

\noindent \textbf{Test case (IV):} (CPU time) to examine the efficiency of the method for the high-dimensional domain, we employ \eqref{testIII},
where $p_1 = 4$, $p_{2i} = 3\frac{1}{3}$, $p_{2i} = 3\frac{2}{7}$, $\epsilon_i=2^{p_{2i}-p_{2i+1}}$, $t \in [0,2]$, and $x \in [-1,1]^d$. In the following numerical examples, we illustrate the convergence rate and efficiency of the method, employing the test cases.

\begin{figure}[pt]
\center
\begin{subfigure}[b]{0.4\textwidth}
\centering
\includegraphics[width=2.0in]{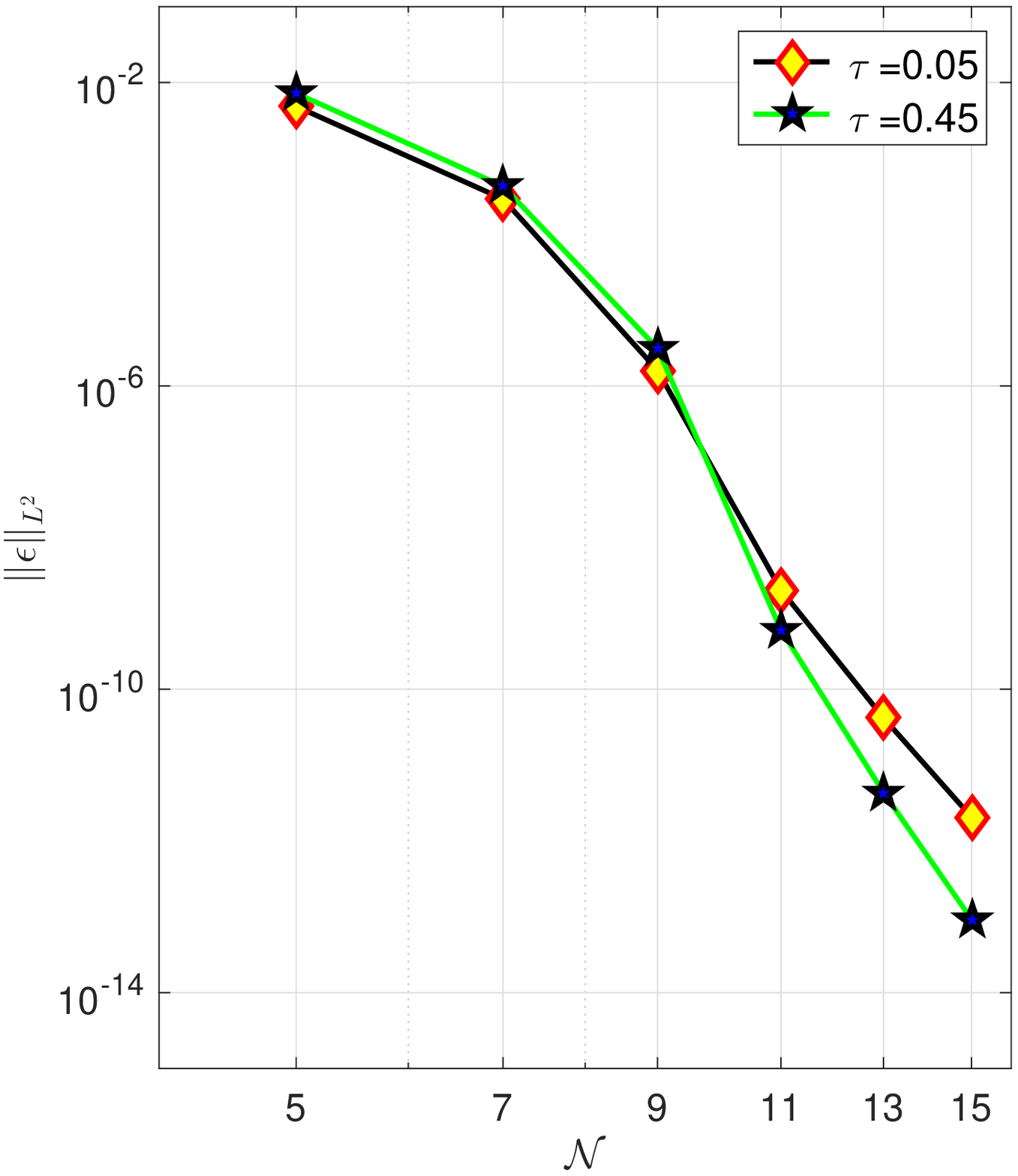}
\caption{}
\end{subfigure}
\begin{subfigure}[b]{0.4\textwidth}
\centering
\includegraphics[width=2.0in]{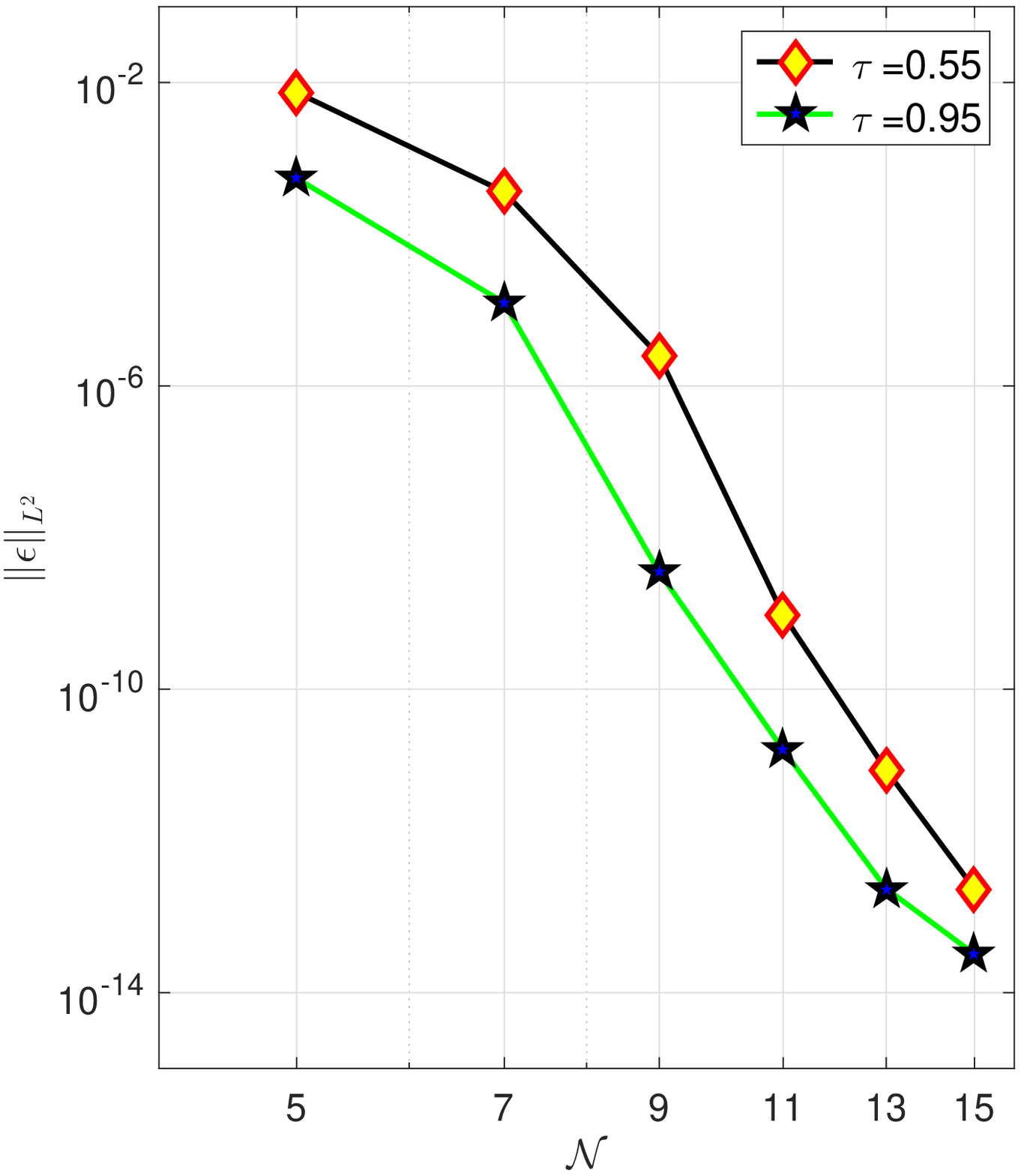}
\caption{}
\end{subfigure}
\caption{
\label{Fig: 1_PG Spectral Method for polynomial 1D-temporal}
Temporal \textit{$p$-refinement}: log-log scale $L^2$-error versus temporal expansion orders $\mathcal{N}$ for test case (I).
}
\end{figure}

\begin{figure}[pt]
\center
\begin{subfigure}[b]{0.4\textwidth}
\centering
\includegraphics[width=2.0in]{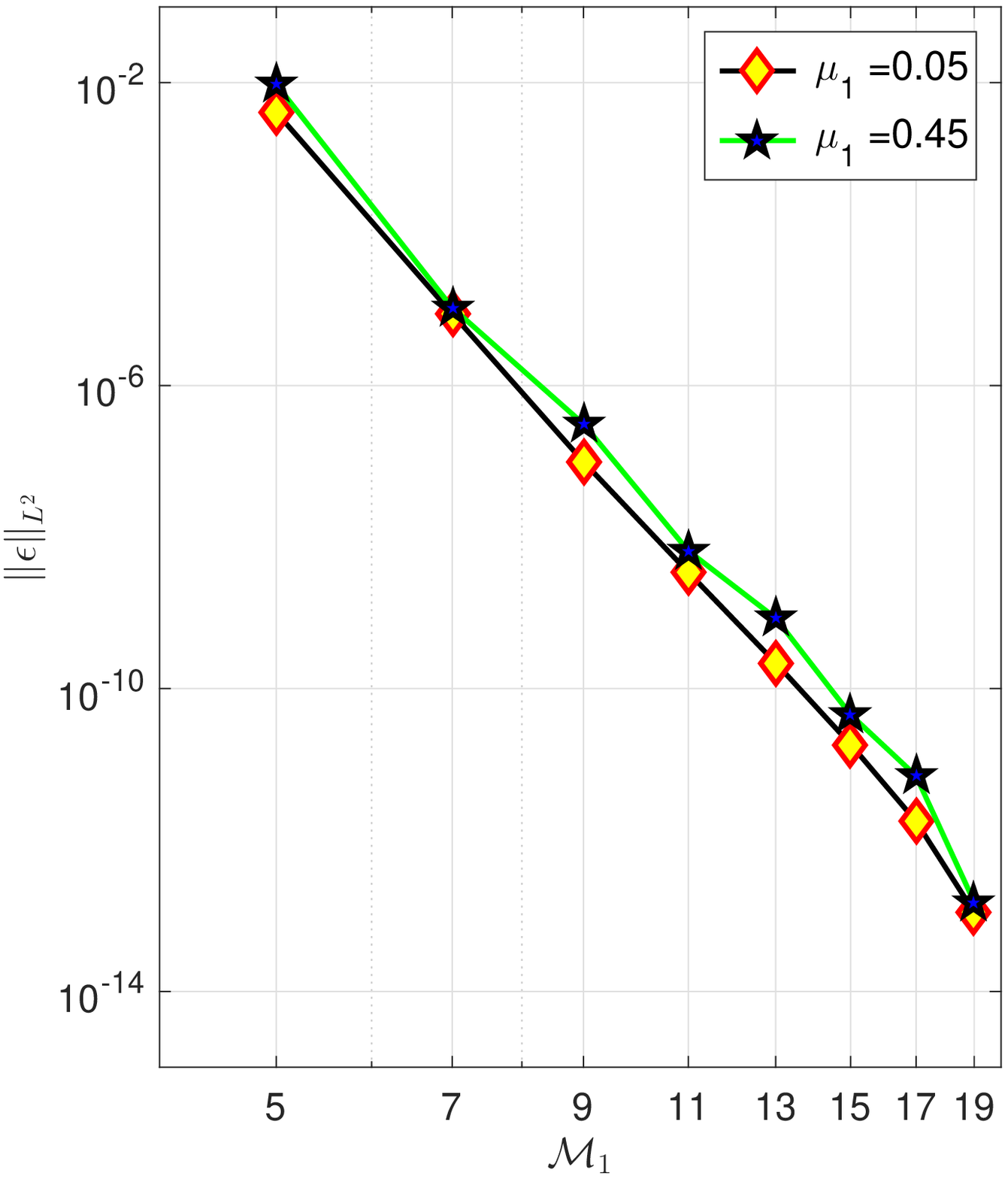}
\caption{}
\end{subfigure}
\begin{subfigure}[b]{0.4\textwidth}
\centering
\includegraphics[width=2.0in]{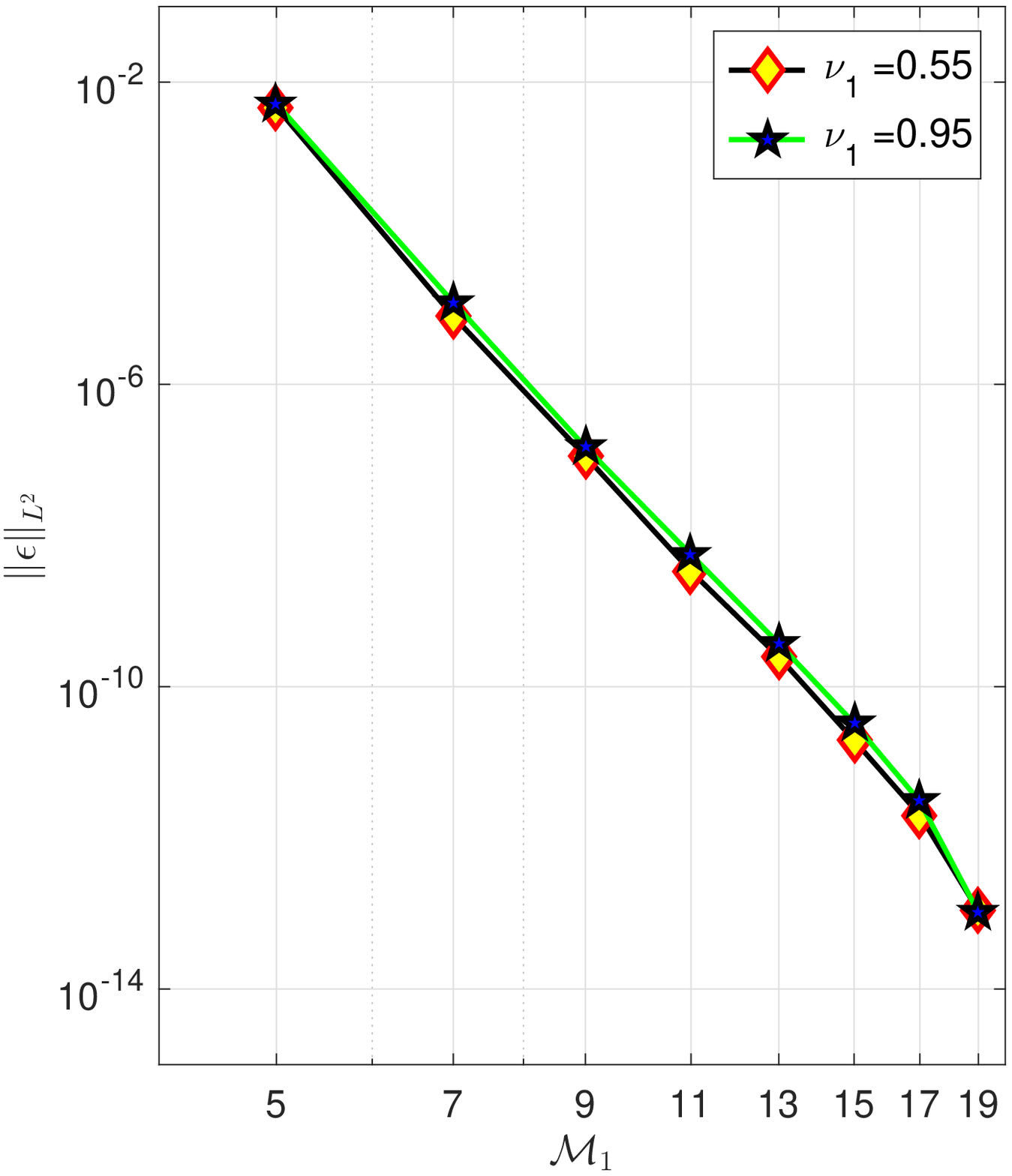}
\caption{}
\end{subfigure}
\caption{\label{Fig: 1_PG Spectral Method for polynomial 1D-spatial}
Spatial \textit{$p$-refinement}: log-log scale $L^2$-error versus spatial expansion orders $\mathcal{M}$ for the test case (I). 
}
\end{figure}

\begin{figure}[pt]
\center
\begin{subfigure}[b]{0.4\textwidth}
\centering
\includegraphics[width=2.0in]{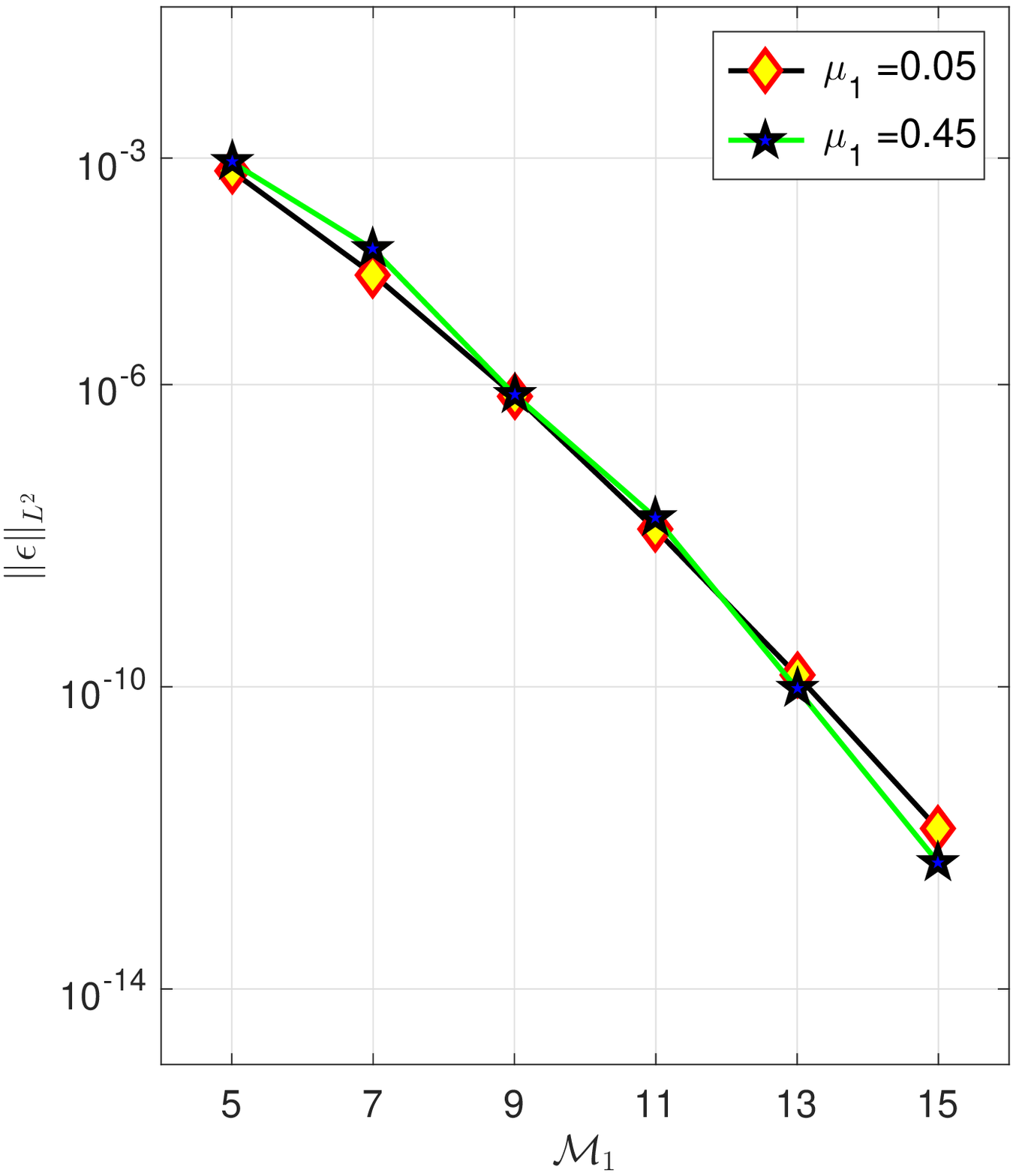}
\caption{}
\end{subfigure}
\begin{subfigure}[b]{0.4\textwidth}
\centering
\includegraphics[width=2.0in]{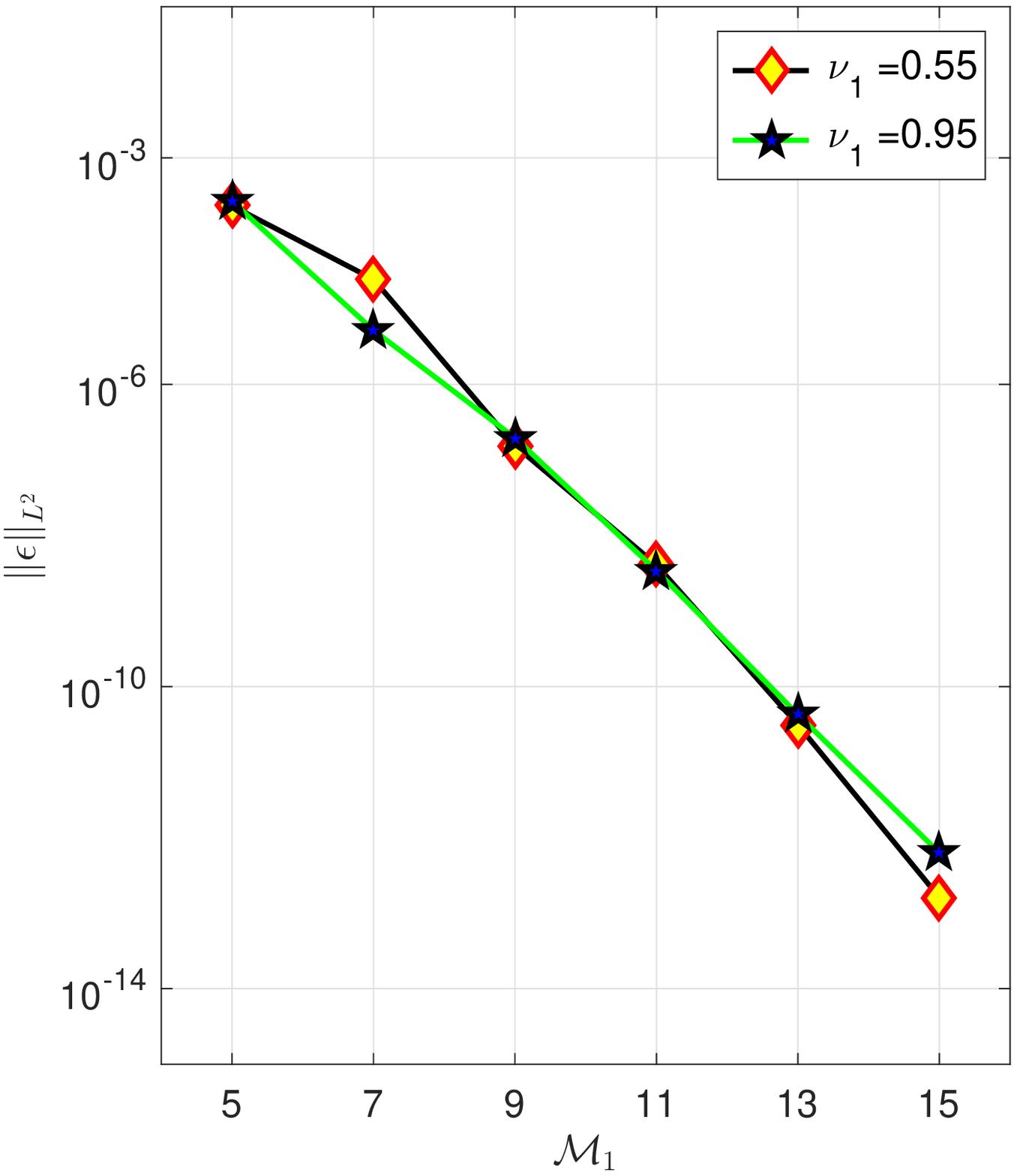}
\caption{}
\end{subfigure}
\caption{\label{Fig: 1_PG Spectral Method for Sinusoidal 1D-spatial}
Exponential convergence in the spatial \textit{$p$-refinement}: log-log scale $L^2$-error versus spatial expansion orders $\mathcal{M}$ for the test case (II). %In the spatial \textit{$p$-refinement} case(a) $T=1$, and case(b) $T=100$ , where $2\tau=\frac{1}{10}, \, 2\tau=\frac{9}{10}$, $2\mu = \frac{5}{10}$, $2\nu = \frac{15}{10}$, and $\mathcal{N}=23$. Here, the exact solution is $u^{exact} = t^{p_1} \times sin[n\pi \, (1+x)]$ where $n=1$ and $P_1 = 6\frac{1}{3}$.
}
\end{figure}

\begin{figure}[h]
\center
\begin{subfigure}[b]{0.4\textwidth}
\centering
\includegraphics[width=2.0in]{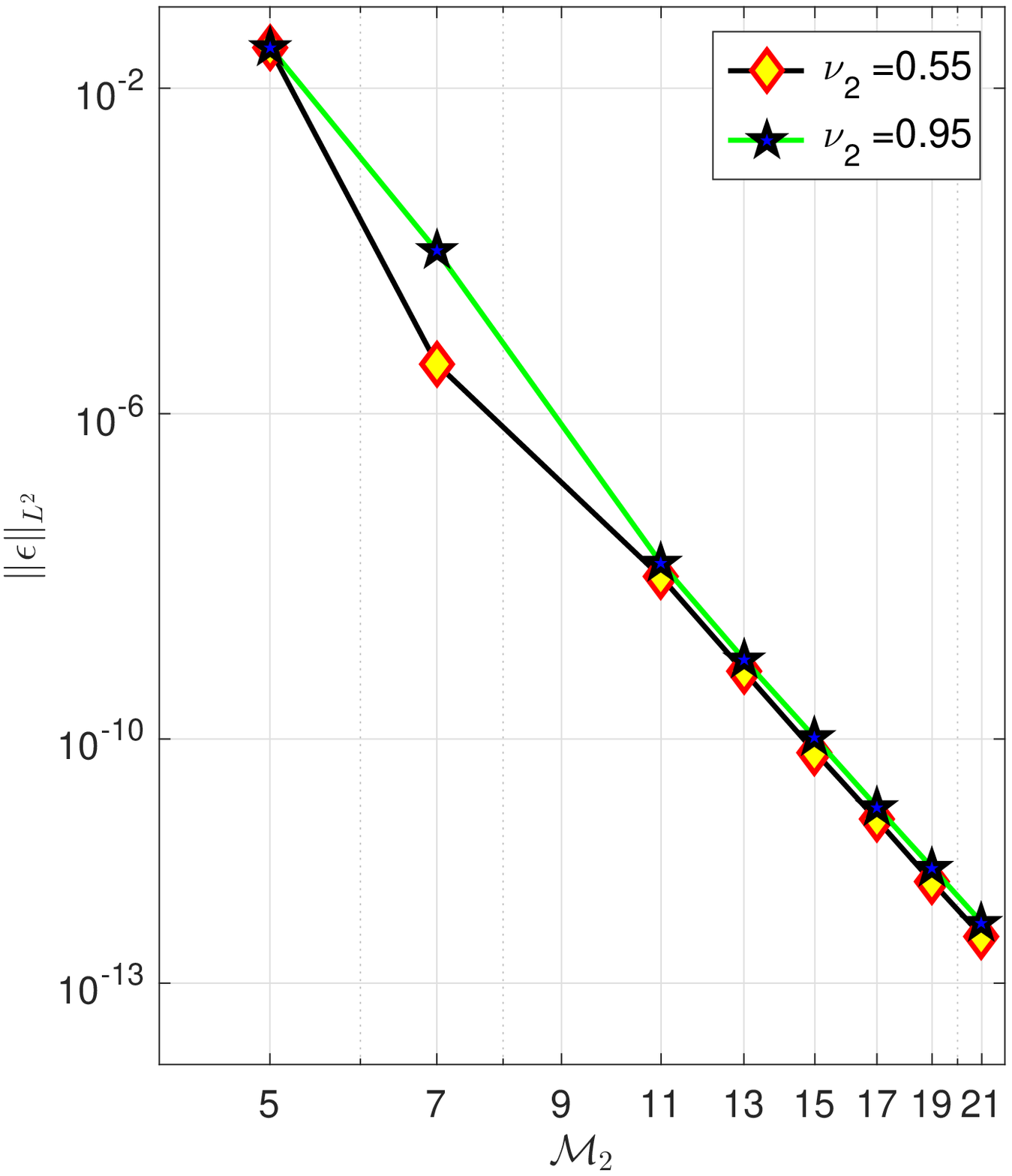}
\caption{}
\end{subfigure}
\begin{subfigure}[b]{0.4\textwidth}
\centering
\includegraphics[width=2.0in]{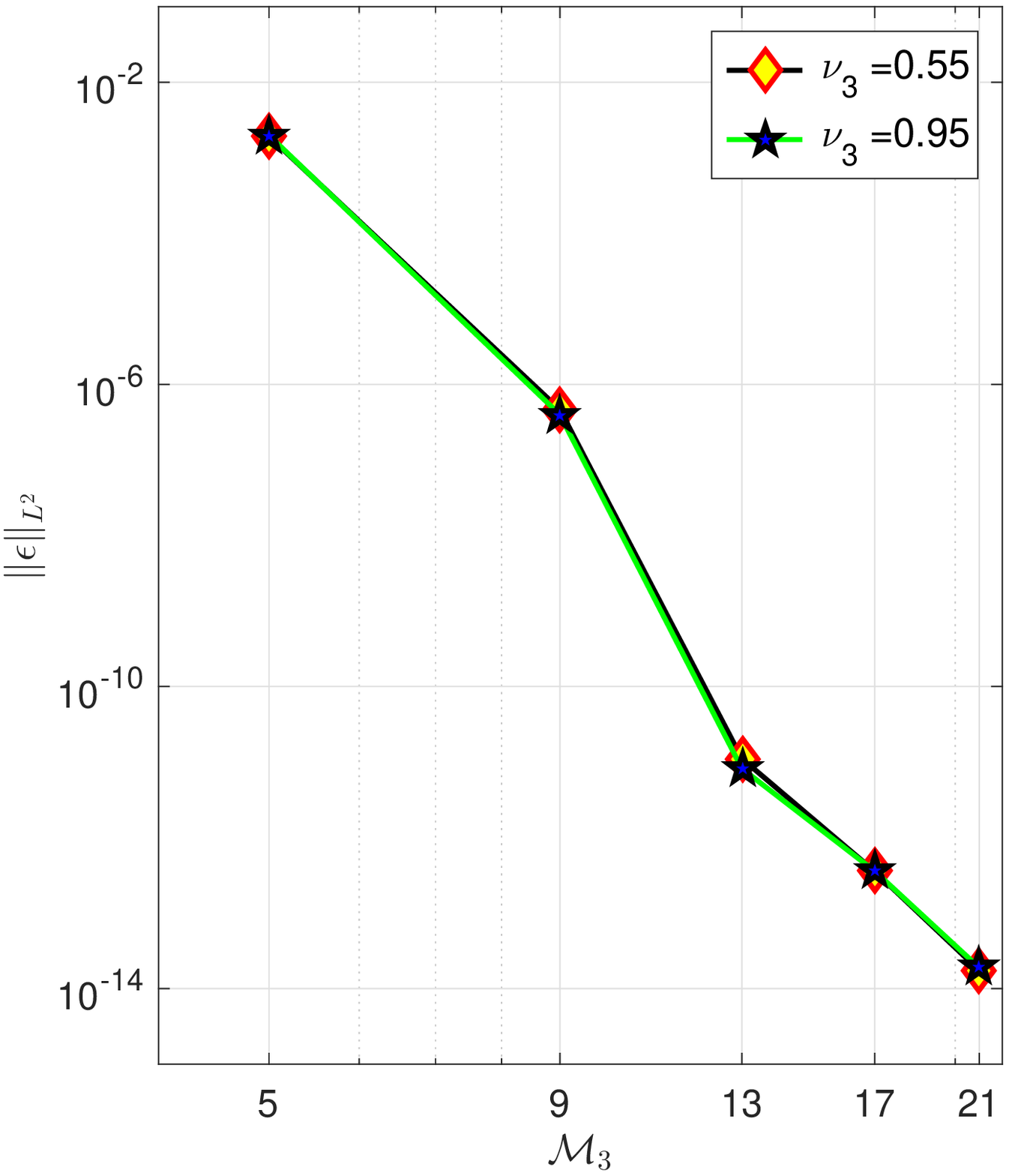}
\caption{}
\end{subfigure}
\caption{\label{Fig: 1_PG Spectral Method for polynomial 2D-spatial}
Spatial \textit{$p$-refinement}: log-log scale $L^{\infty}$-error versus spatial expansion orders $\mathcal{M}_2$, $\mathcal{M}_3$ in the test case (III) for the limit fractional orders of $\nu$.
}
\end{figure}

\begin{figure}[h]
\center
\begin{subfigure}[b]{0.4\textwidth}
\centering
\includegraphics[width=2.0in]{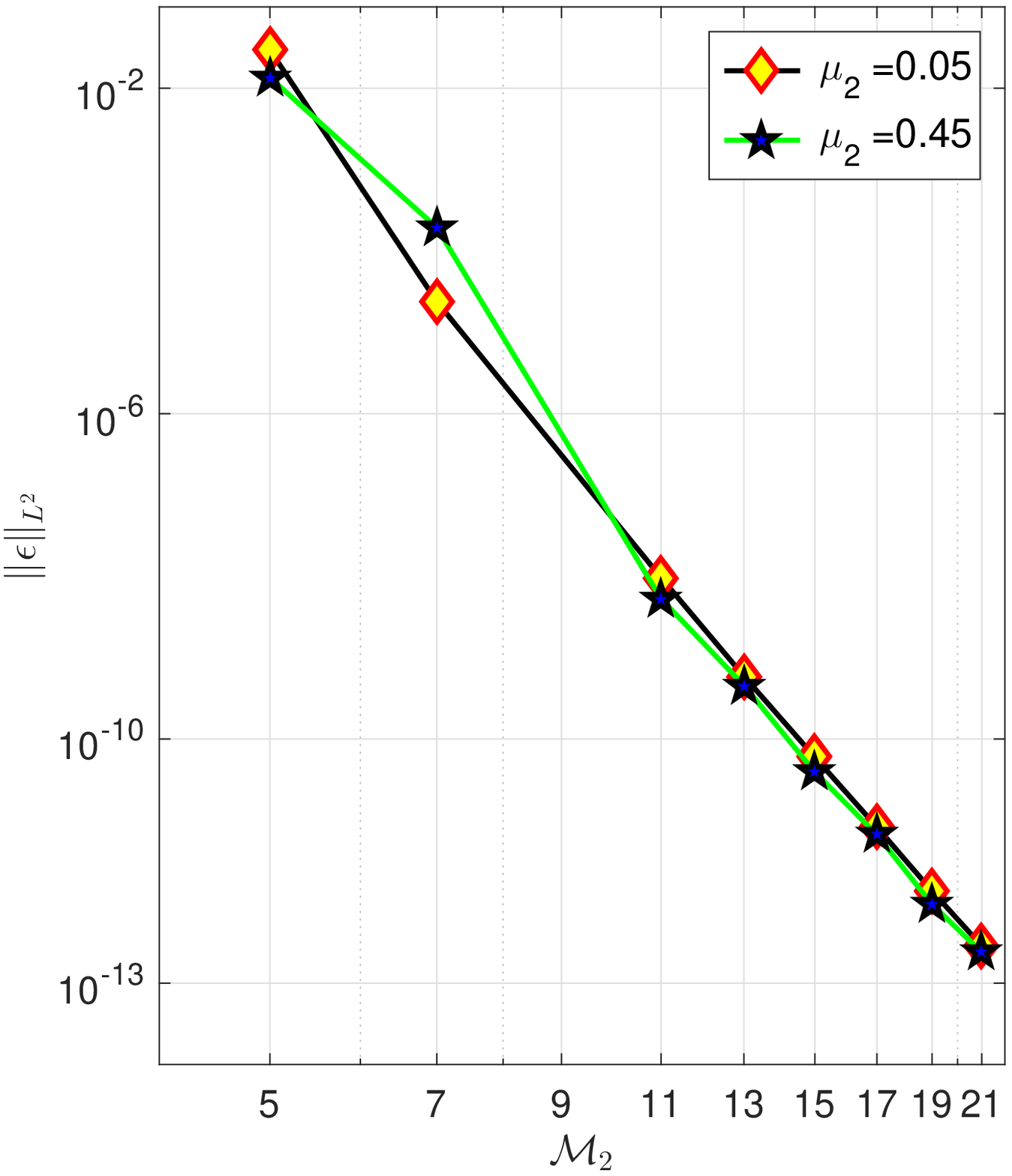}
\caption{}
\end{subfigure}
\begin{subfigure}[b]{0.4\textwidth}
\centering
\includegraphics[width=2.0in]{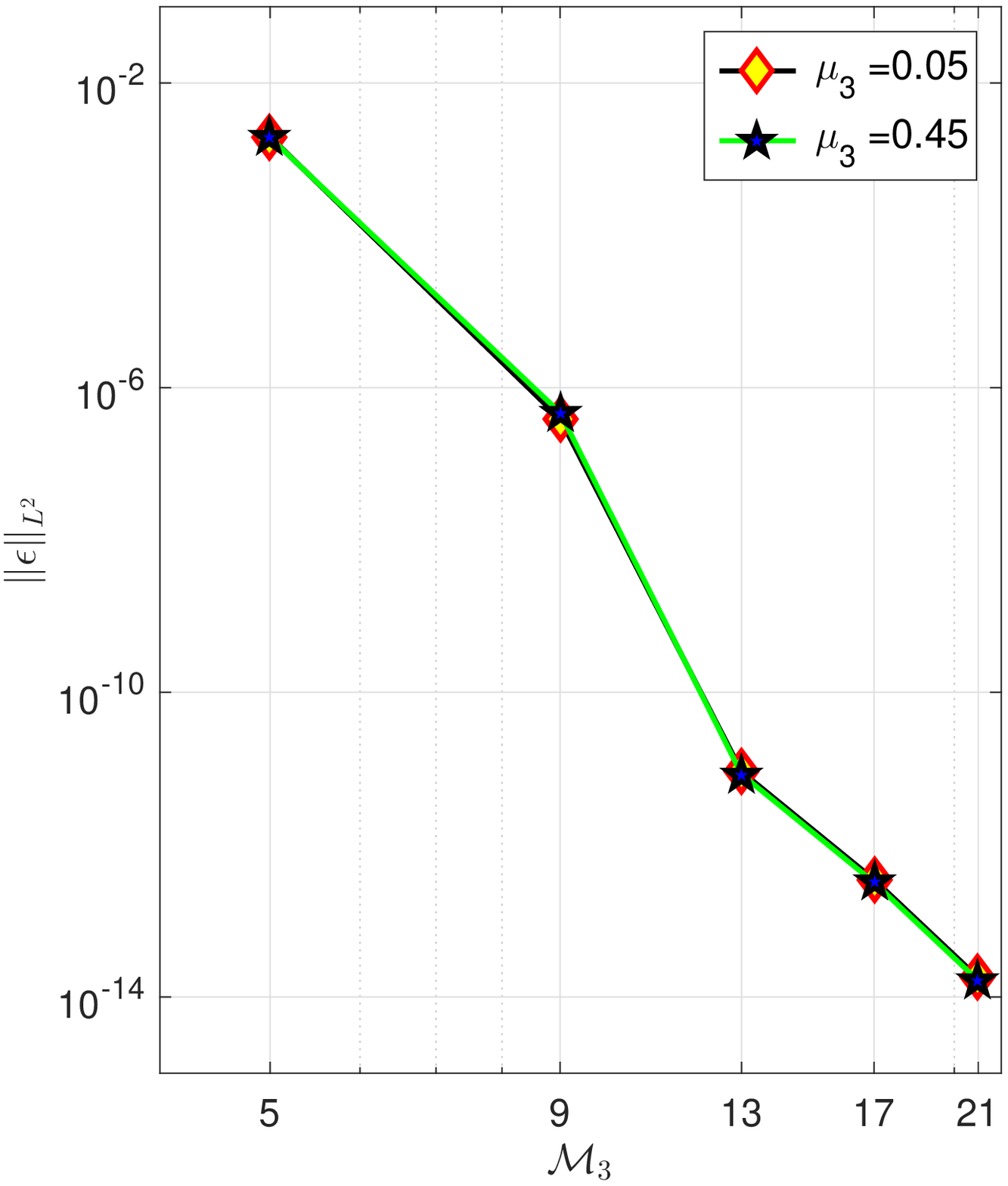}
\caption{}
\end{subfigure}
\caption{\label{Fig: 1_PG Spectral Method for polynomial 3D-spatial}
Spatial \textit{$p$-refinement}: log-log scale $L^{\infty}$-error versus spatial expansion orders $\mathcal{M}_2$, $\mathcal{M}_3$ for the test case (III) for the limit fractional orders of $\mu$.
}
\end{figure}

\subsection{{Numerical Test (I)}}
\label{Sec: Pol_Exact_FPDEs}
We plot the log-log scale $L^2$-error  versus temporal orders $\mathcal{N}$ in Fig. \ref {Fig: 1_PG Spectral Method for polynomial 1D-temporal} in a log-log scale plot for the test case (I) while $2\tau = \frac{1}{10}$, $\frac{9}{10}$, $2\mu_1 = \frac{5}{10}$, $2\nu_1 = \frac{15}{10}$, $T = 2$ and spatial expansion order is fixed ($\mathcal{M}=23$). Having the same set-up, we also consider $2\tau = \frac{11}{10}$, $\frac{19}{10}$ in the temporal direction to examine the spectral convergence of fractional wave equation. The $L^2$-error decays linearly in the log-log scale plot as temporal expansion order $\mathcal{N}$ increases in both cases, indicating the spectral convergence of PG method. In \cite{samiee2016Unified2}, we obtain the theoretical convergence rate of $ \Vert e \Vert_{L^{2}}$ and compare with the corresponding practical ones.     

\subsection{{Numerical Test (II)}}
\label{Sec: Pol_1_FPDEs}
Here, we perform the spatial \textit{$p$-refinement} while the temporal expansion order is fixed for the test case (I). In Fig. \ref{Fig: 1_PG Spectral Method for polynomial 1D-spatial}, spectral convergence of log-log scale $L^{\infty}$-error versus spatial expansion orders $\mathcal{M}_1$ is shown where $2\nu_1 = \frac{11}{10}$, $\frac{19}{10}$ in setup (a). We set $2\tau = \frac{6}{10}$, $2\mu_1 = \frac{5}{10}$, $T = 2$ and temporal expansion order is fixed ($\mathcal{N}=23$). In this case, the limit fractional orders of $\nu_1$ are examined, where both have the spectral convergence but with different rates. We also carried out the spatial \textit{$p$-refinement} for the limit fractional orders of $\mu_1$. The spectral convergence of the PG method is observed, where $2\mu_1 = \frac{1}{10}$, $\frac{9}{10}$ and $2\nu_1 = \frac{15}{10}$. To this end, we can conclude that the PG method in (1+1) dimensional space-time domain is spectrally accurate up to the order of $10^{-15}$.

\subsection{{Numerical Test (III)}}
\label{Sec: sin_Exact_FPDEs}
%Similar to section~\ref{Sec: Pol_High_FPDEs}, we examine (\ref{Eq: St_V2}), while the exact solution in the sinusoidal form is described in the test case (V). 
In Fig.~\ref {Fig: 1_PG Spectral Method for Sinusoidal 1D-spatial}, we plot $ \Vert e \Vert_{L^{2} }  = \Vert u-u^{ext} \Vert_{L^{2}}$ versus spatial expansion orders $\mathcal{M}$ for the test case (II), showing the spatial \textit{$p$-refinement}. In setup (a) $2\nu_1 = \frac{11}{10}$, $ \frac{19}{10}$ and $2\mu_1 = \frac{5}{10}$ and in setup (b) $2\mu_1 = \frac{1}{10}$, $ \frac{9}{10}$ and $2\nu_1 = \frac{15}{10}$ where $2\tau= \frac{6}{10}$. The temporal expansion order ($\mathcal{N}=23$) is fixed. The exponential convergence in the log-linear scale plot is illustrated clearly for the limit fractional orders of $\mu_1$ and $\nu_1$ in case spatial component of the exact solution is a sinusoidal smooth function.

\begin{table}[h]
\center
\caption{\label{Table: Higher-Dimensional FPDEs} Performance study and CPU time (in sec.) of the unified PG spectral method for the test
case (IV). In each step, we uniformly increase the bases order by one in all dimensions.
}
\vspace{0.1in}
%
%\begin{sidewaystable}
%
\begin{tabular}
{l c c}
 & 2-D FADRE  &  \\
\hline
\hline
\multirow{2}{*}{ $\mathcal{N}$} & \multirow{2}{*}{ $\Vert \epsilon \Vert_{L^{\infty} } $}  & CPU Time \\
&&[Sec] \\
[3pt]
\hline
$5$ &  0.008 & 1.48 \\
\hline
$7$ &  0.0003 & 3.01 \\
\hline
$9$ &  1.69$\times10^{-6}$ & 3.48 \\
\hline
$15$ &  2.96$\times10^{-11}$ & 4.95 \\
\hline
\vspace{0.2in}
 \end{tabular}
\hspace{0.4in}
\begin{tabular}
{l c c}
%{c@{\hspace*{.30in}}c@{\hspace*{.30in}}c@{\hspace*{.30in}}c@{\hspace*{.30in}}c@{\hspace*{.30in}}c} 
 & 3-D FADRE  & \\
\hline
\hline
\multirow{2}{*}{ $\mathcal{N}$} & \multirow{2}{*}{ $\Vert \epsilon \Vert_{L^{\infty} } $}  & CPU Time  \\&&[Sec] \\[3pt]
\hline
$5$ &  0.01 & 1.43 \\
\hline
$7$ &  0.0003 & 5.39 \\
\hline
$9$ &  2.6$\times10^{-7}$ & 6.14 \\
\hline
$15$ &  2.41$\times10^{-10}$ & 7.54 \\
\hline
\vspace{0.2in}
 \end{tabular}
 \vspace{0.05in}
\begin{tabular}
{l c c}
%
%{c@{\hspace*{.30in}}c@{\hspace*{.30in}}c@{\hspace*{.30in}}c@{\hspace*{.30in}}c@{\hspace*{.30in}}c} 
 & 4-D FADRE  &  \\ 
\hline
\hline
\multirow{2}{*}{ $\mathcal{N}$} & \multirow{2}{*}{ $\Vert \epsilon \Vert_{L^{\infty} } $}  & CPU Time  \\
&&[Sec]
\\[3pt]
\hline
$5$ &  0.00005 & 3.56 \\
\hline
$7$ & 3.31$\times10^{-7}$ &  8.87\\
\hline
$9$ &  8.17$\times10^{-9}$ &  5.37 \\
\hline
$15$ &  9.70$\times10^{-12}$ &  55.78 
\\
\hline
 \end{tabular}
 \end{table}

\subsection{{Numerical Test (IV)}}
\label{Sec: Pol_High_FPDEs}

In addition to spatial/temporal \textit{$p$-refinement}, we perform \textit{$p$-refinement} for (1+2) and (1+3) as the higher dimensional domain in the test case (III). In Fig. \ref{Fig: 1_PG Spectral Method for polynomial 2D-spatial}, the spectral convergence of log-log $L^{\infty}$-error versus spatial expansion orders $\mathcal{M}_2$, $\mathcal{M}_3$  is shown. In setup (a),  $2\nu_2 = \frac{11}{10}$, $\frac{19}{10}$ while $2\nu_1 = \frac{15}{10}$, $2\mu_1 = \frac{4}{10}$ and $2\mu_2 = \frac{6}{10}$ and setup (b) $2\nu_3 = \frac{11}{10}$, $\frac{19}{10}$ while $2\nu_1 = \frac{14}{10}$, $2\nu_2 = \frac{16}{10}$, $2\mu_1 = \frac{3}{10}$, $2\mu_2 = \frac{5}{10}$ and $2\mu_3 = \frac{7}{10}$, where $2\tau=\frac{6}{10}$, $T=2$. Furthermore, we increase the maximum bases order uniformly in all dimensions. 

Similarly, we perform the spatial \textit{$p$-refinement} for the limit fractional orders of $\mu$ in FADE. We study setup (a) $2\mu_2 = \frac{1}{10}$, $\frac{9}{10}$ while $2\mu_1 = \frac{5}{10}$, and setup (b) $2\mu_3 = \frac{1}{10}$, $\frac{9}{10}$ while $2\mu_1 = \frac{4}{10}$, $2\mu_2 = \frac{6}{10}$. In both setups, $2\tau=\frac{6}{10}$, $2\nu_1=\frac{15}{10}$, $2\nu_2=\frac{15}{10}$, $T=2$. Furthermore, $\mathcal{N}=\mathcal{M}_1=\mathcal{M}_2=\mathcal{M}_3$ changes concurrently. In Fig. \ref{Fig: 1_PG Spectral Method for polynomial 3D-spatial}, the PG method shows spectral convergence for the limit fractional orders of $\mu$.  
%It is noteworthy that the study of \textit{$p$-refinement} for fractional orders of $\mu$ is meaningless while higher orders such as $\nu$ exists in FADE.

\subsection{{Numerical Test (V)}}
\label{Sec: Pol_High_FPDEs_2}

To examine the efficiency of the PG method and the fast solver in high-dimensional problem, the convergence results and CPU time for test case (IV) are presented in Table \ref{Table: Higher-Dimensional FPDEs} for (1+1), (1+3) and (1+5) dimensional space-time hypercube domains where the error is measured by the essential norm $\Vert e \Vert_{\infty}$ in the test case (IV). The CPU time is obtained on {Intel (Xeon E52670) 2.5 GHz processor.} The presented PG method remains spectrally accurate in (1+5) dimensional time-space domain.

%\newpage

\section{Summary and Discussion}

%\subsection{Summary}
%\label{Summary}

We developed a new unified Petrov-Galerkin spectral method for a class of fractional partial differential equations with constant coefficients (\ref{Eq: Genral_FPDE}) in a ($1+d$)-dimensional \textit{space-time} hypercube, $d = 1, 2, 3$, etc, subject to homogeneous Dirichlet initial/boundary conditions. We employed \textit{Jacobi poly-fractonomial}s, as temporal basis/test functions, and the Legendre polynomials as spatial basis/test functions, yielding spatial mass matrices being independent of the spatial fractional orders. Additionally, we formulated the novel unified fast linear solver for the resulting high-dimensional linear system, which reduces the computational cost significantly. In fact, the main idea of the paper was to formulate a closed-form solution for the high-dimensional Lyapunov equation in terms of the eigensolutions up to the precision accuracy of computationally obtained eigensolutions.
%based on the eigen-solution of the generalized eigen-problem of spatial mass matrices with respect to the corresponding stiffness matrices in each dimension, hence, reducing the computational cost significantly, $\mathcal{O}(N^{d+2})$. 
The PG method has been illustrated to be spectrally accurate for power-law test cases in each dimension. Furthermore, exponential convergence is observed for a sinusoidal smooth function in a spatial \textit{p}-refinement. 
%The fast solver significantly reduces the computational cost particularly in higher-dimensional problems. We examined the efficiency of the method through the CPU time and rate of convergence in (1+1), (1+3), and (1+5) space-time dimensional domain. 
To check the stability and spectral convergence of the PG method, we carried out the corresponding discrete stability and error analysis of the method for (\ref{Eq: con-dim PG method}) in \cite{samiee2016Unified2}.
Despite the high accuracy and the efficiency of the method especially in higher-dimensional problems, treatment of FPDEs in complex geometries and FPDEs with variable coefficients will be studies in our future works.

\section*{Appendix}
\label{Sec: Appendix}

Here, we provide the force function based on the exact solutions.

\subsection*{$\bullet$ \textbf{Force term of test case (I)}}
\label{lemma_f_sin}
To obtain $f$ in \eqref{Eq: Genral_FPDE} based on \eqref{test11111}, first we need to calculate all fractional derivatives of $u^{ext}$. To satisfy the corresponding boundary conditions, $\epsilon_i = 2^{p_{2i}-p_{2i+1}}$. 
Take $X^{T} = t^{p_1}$ and $X_i^{S} = (1+\zeta_i)^{p_{2i}}-\epsilon_i \, (1+\zeta_i)^{p_{2i+1}}$, where $\zeta_i = 2\frac{x_i-a_i}{b_i-a_i}-1$ and $\zeta_i \, \in \, [-1\, , \, 1]$.
Considering (\ref{eq2}),
\begin{eqnarray}
\label{Eq: f_t_derivative}
&\prescript{}{0}{\mathcal{D}}_{t}^{2\tau} \, X^{T} = \frac{\Gamma[p_1+1]}{\Gamma[p_1+1-2\tau]} \, t^{p_1-2\tau} = (\frac{T}{2})^{p_1-2\tau}\frac{\Gamma[p_1+1]}{\Gamma[p_1+1-2\tau]} \, (1+\eta(t))^{p_1-2\tau}, & \quad
\end{eqnarray}
where $\eta (t) = 2(\frac{t}{T})-1$. Similarly,
\begin{eqnarray}
\label{Eq: f_s1_derivative}
&\prescript{}{a_i}{\mathcal{D}}_{x_i}^{2\mu_i} \, X_{i}^{S} = \Big{(} \frac{b_i-a_i}{2}  \Big{)}^{-2\mu_i} \Big{[} \frac{\Gamma[p_{2i}+1]}{\Gamma[p_{2i}+1-2\mu_{i}]}(1+\zeta_{2i}(x_i))^{p_{2i}-2\mu_{i}} \,- &
\nonumber
\\
 &\epsilon_{i} \frac{\Gamma[p_{2i+1}+1]}{\Gamma[p_{2i+1}+1-2\mu_{i}]}(1+\zeta_{2i}(x_i))^{p_{2i+1}-2\mu_{i}}\Big{]},& \quad
\end{eqnarray}
and
\begin{eqnarray}
\label{Eq: f_s1_derivative_disp}
&\prescript{}{a_i}{\mathcal{D}}_{x_i}^{2\nu_i} \, X_{i}^{S} = \Big{(} \frac{b_i-a_i}{2}  \Big{)}^{-2\nu_i -2} \Big{[} \frac{\Gamma[p_{2i}+1]}{\Gamma[p_{2i}+1-2\nu_{i}]}(1+\zeta_{2i}(x_i))^{p_{2i}-2\nu_{i}} \,- &
\nonumber
\\
 &\epsilon_{i} \frac{\Gamma[p_{2i+1}+1]}{\Gamma[p_{2i+1}+1-2\nu_{i}]}(1+\zeta_{2i}(x_i))^{p_{2i+1}-2\nu_{i}}\Big{]}.& \quad
\end{eqnarray}
Therefore,
\begin{eqnarray}
\label{Eg: f representation_222}
f &=& (\frac{T}{2})^{p_1-2\tau}\frac{\Gamma[p_1+1]}{\Gamma[p_1+1-2\tau]} \, (1+\eta)^{p_1-2\tau} \prod_{i=1}^{d} (1+\zeta_i)^{p_{2i}}-\epsilon_i \, (1+\zeta_i)^{p_{2i+1}} 
\nonumber
\\
&+& \sum_{i=1}(\frac{T}{2})^{p_1}(1+\eta)^{p_1} \Big{(}c_{l_i} \, \Big{(} \frac{b_i-a_i}{2}  \Big{)}^{-2\mu_i} \Big{[} \frac{\Gamma[p_{2i}+1]}{\Gamma[p_{2i}+1-2\mu_{i}]}(1+\zeta_{2i})^{p_{2i}-2\mu_{i}} \,
\nonumber
\\
&-&\epsilon_{i} \frac{\Gamma[p_{2i+1}+1]}{\Gamma[p_{2i+1}+1-2\mu_{i}]}(1+\zeta_{2i})^{p_{2i+1}-2\mu_{i}}\Big{]}\, 
\prod_{j=1, \, j\neq i}^{d}[(1+\zeta_j)^{p_{2j}}-\epsilon_j \, (1+\zeta_j)^{p_{2j+1}}]
%&
%\nonumber
%\\
%&
%-c_{r_i} \, \Big{(} \frac{b_i-a_i}{2}  \Big{)}^{-2\mu_i} \Big{[} \frac{\Gamma[p_{2i}+1]}{\Gamma[p_{2i}+1-2\mu_{i}]}(1+\zeta_{2i})^{p_{2i}-2\mu_{i}} \,- 
%\epsilon_{i} \times &
 %\nonumber
%\\
%& \frac{\Gamma[p_{2i+1}+1]}{\Gamma[p_{2i+1}+1-2\mu_{i}]}(1+\zeta_{2i})^{p_{2i+1}-2\mu_{i}}\Big{]} \prod_{j=1, \, j\neq i}^{d}[(1+\zeta_j)^{p_{2j}}-\epsilon_j \, (1+\zeta_j)^{p_{2j+1}}]
\Big{)} 
\nonumber
\\
&-& \sum_{i=1}(\frac{T}{2})^{p_1}(1+\eta)^{p_1} \Big{(}\kappa_{l_i} \, \Big{(} \frac{b_i-a_i}{2}  \Big{)}^{-2\nu_i-2} \Big{[} \frac{\Gamma[p_{2i}+1]}{\Gamma[p_{2i}+1-2\mu_{i}]}(1+\zeta_{2i})^{p_{2i}-2\nu_{i}} \,
\nonumber
\\
&-& \epsilon_{i} \frac{\Gamma[p_{2i+1}+1]}{\Gamma[p_{2i+1}+1-2\nu_{i}]}(1+\zeta_{2i})^{p_{2i+1}-2\nu_{i}}\Big{]}\, \prod_{j=1, \, j\neq i}^{d}[(1+\zeta_j)^{p_{2j}}-
\epsilon_j \, (1+\zeta_j)^{p_{2j+1}}]
\Big{)}. \quad \quad
\end{eqnarray}

%\end{proof}

 \subsection*{$\bullet$ \textbf{Force term of test case (II)}}
\label{lemma_f_pwr}

Take $X^{T} = t^{p_1}$ and $X_i^{S} = sin\big{(}n \pi \zeta\big{)}$. Here, we approximate $X_i^{S}$ as
\begin{equation}
X^{S} = \Sigma_{j=1}^{N_s} (-1)^{2j-1}\frac{(n \pi \zeta)^{2j-1}}{(2j-1)!},
\end{equation}
where $N_s$ controls the level of approximation error. Taking the same steps of \eqref{Eg: f representation_222}, we obtain
\begin{eqnarray}
\label{Eg: f representation_223}
f &=& (\frac{T}{2})^{p_1-2\tau}\frac{\Gamma[p_1+1]}{\Gamma[p_1+1-2\tau]} \, (1+\eta)^{p_1-2\tau} \Sigma_{j=1}^{N_s} (-1)^{2j-1}\frac{(n \pi \zeta)^{2j-1}}{(2j-1)!} 
\nonumber
\\
&+& (\frac{T}{2})^{p_1}(1+\eta)^{p_1} \, \big{[}(c_l) \, \Big{(} \frac{b-a}{2}  \Big{)}^{-2\mu} \,  \Sigma_{j=1}^{N_s} (-1)^{2j-1}\frac{(n \pi \zeta)^{2j-1}}{(2j-1)!} \frac{\Gamma[2j]}{\Gamma[2j-2\mu]} \, \zeta ^{2j-2\mu}\, 
\nonumber
\\
& - &\, (\kappa_l) \, \Big{(} \frac{b-a}{2}  \Big{)}^{-2\nu-2} \,  \Sigma_{j=1}^{N_s} (-1)^{2j-1}\frac{(n \pi \zeta)^{2j-1}}{(2j-1)!}\frac{\Gamma[2j]}{\Gamma[2j-2\nu]} \, \zeta ^{2j-2\nu}\big{]}.\, \,
\end{eqnarray}

\section*{Acknowledgement}
This work was supported by the AFOSR Young Investigator Program (YIP) award on: “Data-Infused Fractional PDE Modelling and Simulation of Anomalous Transport” (FA9550-17-1-0150) and by the MURI/ARO on Fractional PDEs for Conservation Laws and Beyond: Theory, Numerics and Applications (W911NF- 15-1-0562).

%\center{\textbf{References}}
%\bibliographystyle{elsarticle-num-names}
%\bibliographystyle{elsarticle-num}
\bibliographystyle{siam}
\bibliography{RFSLP_Refs2}

\begin{thebibliography}{10}

\bibitem{baeumer2016space}
{\sc Boris Baeumer, Tomasz Luks, and Mark~M Meerschaert}, {\em Space-time
  fractional dirichlet problems}, arXiv preprint arXiv:1604.06421,  (2016).

\bibitem{benson2001fractional}
{\sc David~A Benson, Rina Schumer, Mark~M Meerschaert, and Stephen~W
  Wheatcraft}, {\em Fractional dispersion, l{\'e}vy motion, and the made tracer
  tests}, in Dispersion in Heterogeneous Geological Formations, Springer, 2001,
  pp.~211--240.

\bibitem{Benson2000}
{\sc D.~A. Benson, S.~W. Wheatcraft, and M.~M. Meerschaert}, {\em Application
  of a fractional advection-dispersion equation}, Water Resources Research, 36
  (2000), pp.~1403--1412.

\bibitem{carpinteri2014fractals}
{\sc Alberto Carpinteri and Francesco Mainardi}, {\em Fractals and fractional
  calculus in continuum mechanics}, vol.~378, Springer, 2014.

\bibitem{chen2015multi}
{\sc Feng Chen, Qinwu Xu, and Jan~S Hesthaven}, {\em A multi-domain spectral
  method for time-fractional differential equations}, Journal of Computational
  Physics, 293 (2015), pp.~157--172.

\bibitem{chen2014second}
{\sc Minghua Chen and Weihua Deng}, {\em A second-order numerical method for
  two-dimensional two-sided space fractional convection diffusion equation},
  Applied Mathematical Modelling, 38 (2014), pp.~3244--3259.

\bibitem{chen2015generalized}
{\sc Sheng Chen, Jie Shen, and Li-Lian Wang}, {\em Generalized jacobi functions
  and their applications to fractional differential equations}, Preprint on
  arXiv,  (2015).

\bibitem{chen2012space}
{\sc Zhen-Qing Chen, Mark~M Meerschaert, and Erkan Nane}, {\em Space--time
  fractional diffusion on bounded domains}, Journal of Mathematical Analysis
  and Applications, 393 (2012), pp.~479--488.

\bibitem{dehghan2016analysis}
{\sc Mehdi Dehghan, Mostafa Abbaszadeh, and Akbar Mohebbi}, {\em Analysis of
  two methods based on galerkin weak form for fractional diffusion-wave:
  Meshless interpolating element free galerkin (iefg) and finite element
  methods}, Engineering Analysis with Boundary Elements, 64 (2016),
  pp.~205--221.

\bibitem{dehghan2016use}
\leavevmode\vrule height 2pt depth -1.6pt width 23pt, {\em The use of element
  free galerkin method based on moving kriging and radial point interpolation
  techniques for solving some types of turing models}, Engineering Analysis
  with Boundary Elements, 62 (2016), pp.~93--111.

\bibitem{del2004fractional}
{\sc Diego del Castillo-Negrete, BA~Carreras, and VE~Lynch}, {\em Fractional
  diffusion in plasma turbulence}, Physics of Plasmas (1994-present), 11
  (2004), pp.~3854--3864.

\bibitem{del1993chaotic}
{\sc Diego del Castillo-Negrete and PJ~Morrison}, {\em Chaotic transport by
  rossby waves in shear flow}, Physics of Fluids A: Fluid Dynamics (1989-1993),
  5 (1993), pp.~948--965.

\bibitem{feng2016high}
{\sc LB~Feng, P~Zhuang, F~Liu, I~Turner, and J~Li}, {\em High-order numerical
  methods for the riesz space fractional advection--dispersion equations},
  Computers \& Mathematics with Applications,  (2016).

\bibitem{guo2015fractional}
{\sc Boling Guo, Xueke Pu, and Fenghui Huang}, {\em Fractional partial
  differential equations and their numerical solutions}, World Scientific,
  2015.

\bibitem{hejazi2013finite}
{\sc Hala Hejazi, Timothy Moroney, and Fawang Liu}, {\em A finite volume method
  for solving the two-sided time-space fractional advection-dispersion
  equation}, Open Physics, 11 (2013), pp.~1275--1283.

\bibitem{kharazmi2016petrov}
{\sc Ehsan Kharazmi, Mohsen Zayernouri, and George~Em Karniadakis}, {\em A
  petrov-galerkin spectral element method for fractional elliptic problems},
  arXiv preprint arXiv:1610.08608,  (2016).

\bibitem{kharazmi2017petrov}
\leavevmode\vrule height 2pt depth -1.6pt width 23pt, {\em Petrov--galerkin and
  spectral collocation methods for distributed order differential equations},
  SIAM Journal on Scientific Computing, 39 (2017), pp.~A1003--A1037.

\bibitem{Klages2008}
{\sc R.~Klages, G.~Radons, and I.~M. Sokolov}, {\em Anomalous Transport:
  Foundations and Applications}, Wiley-VCH, 2008.

\bibitem{li2016linear}
{\sc Dongfang Li, Chengjian Zhang, and Maohua Ran}, {\em A linear finite
  difference scheme for generalized time fractional burgers equation}, Applied
  Mathematical Modelling,  (2016).

\bibitem{li2010existence}
{\sc Xianjuan Li and Chuanju Xu}, {\em Existence and uniqueness of the weak
  solution of the space-time fractional diffusion equation and a spectral
  method approximation},  (2010).

\bibitem{Lischke2017petrov}
{\sc Anna Lischke, Mohsen Zayernouri, and George~Em Karniadakis}, {\em A
  petrov--galerkin spectral method of linear complexity for fractional
  multiterm odes on the half line}, SIAM Journal on Scientific Computing, 39
  (2017), pp.~A1922--A946.

\bibitem{lubich1986discretized}
{\sc Ch~Lubich}, {\em Discretized fractional calculus}, SIAM Journal on
  Mathematical Analysis, 17 (1986), pp.~704--719.

\bibitem{magin2006fractional}
{\sc Richard~L Magin}, {\em Fractional calculus in bioengineering}, Begell
  House Redding, 2006.

\bibitem{mainardi2010fractional}
{\sc Francesco Mainardi}, {\em Fractional calculus and waves in linear
  viscoelasticity: an introduction to mathematical models}, World Scientific,
  2010.

\bibitem{mao2016efficient}
{\sc Zhiping Mao and Jie Shen}, {\em Efficient spectral--galerkin methods for
  fractional partial differential equations with variable coefficients},
  Journal of Computational Physics, 307 (2016), pp.~243--261.

\bibitem{meerschaert2014tempered}
{\sc Mark~M Meerschaert, Farzad Sabzikar, Mantha~S Phanikumar, and Aklilu
  Zeleke}, {\em Tempered fractional time series model for turbulence in
  geophysical flows}, Journal of Statistical Mechanics: Theory and Experiment,
  2014 (2014), p.~P09023.

\bibitem{meerschaert2012stochastic}
{\sc Mark~M Meerschaert and Alla Sikorskii}, {\em Stochastic models for
  fractional calculus}, vol.~43, Walter de Gruyter, 2012.

\bibitem{meerschaert2004finite}
{\sc Mark~M Meerschaert and Charles Tadjeran}, {\em Finite difference
  approximations for fractional advection--dispersion flow equations}, Journal
  of Computational and Applied Mathematics, 172 (2004), pp.~65--77.

\bibitem{metzler2000random}
{\sc Ralf Metzler and Joseph Klafter}, {\em The random walk's guide to
  anomalous diffusion: a fractional dynamics approach}, Physics reports, 339
  (2000), pp.~1--77.

\bibitem{mokhtary2015discrete}
{\sc P~Mokhtary}, {\em Discrete galerkin method for fractional
  integro-differential equations}, arXiv preprint arXiv:1501.01111,  (2015).

\bibitem{naghibolhosseini2015estimation}
{\sc M.~Naghibolhosseini}, {\em Estimation of outer-middle ear transmission
  using \uppercase{DPOAE}s and fractional-order modeling of human middle ear},
  PhD thesis, City University of New York, NY., 2015.

\bibitem{perdikaris2014fractional}
{\sc Paris Perdikaris and George~Em Karniadakis}, {\em Fractional-order
  viscoelasticity in one-dimensional blood flow models}, Annals of biomedical
  engineering, 42 (2014), pp.~1012--1023.

\bibitem{podlubny1998fractional}
{\sc Igor Podlubny}, {\em Fractional differential equations: an introduction to
  fractional derivatives, fractional differential equations, to methods of
  their solution and some of their applications}, vol.~198, Academic press,
  1998.

\bibitem{Askey1969Integral}
{\sc J.~Fitch R.~Askey}, {\em Integral representations for jacobi polynomials
  and some applications}, Journal of Mathematical Analysis and Applications, 26
  (1969).

\bibitem{regner2014randomness}
{\sc Benjamin~Michael Regner}, {\em Randomness in biological transport},
  (2014).

\bibitem{samiee2016Unified2}
{\sc Mehdi Samiee, Mohsen Zayernouri, and Mark~M. Meerschaert}, {\em A unified
  spectral method for fpdes with two-sided derivatives; part ii: Stability and
  error analysis}, submitted to Journal of Computational Physics, 2016,
  (2016).

\bibitem{samko1993fractional}
{\sc Stefan~G Samko, Anatoly~A Kilbas, and Oleg~I Marichev}, {\em Fractional
  integrals and derivatives}, Theory and Applications, Gordon and Breach,
  Yverdon, 1993 (1993).

\bibitem{shen2007fourierization}
{\sc Jie Shen and Li-Lian Wang}, {\em Fourierization of the legendre--galerkin
  method and a new space--time spectral method}, Applied numerical mathematics,
  57 (2007), pp.~710--720.

\bibitem{solomon1993observation}
{\sc TH~Solomon, Eric~R Weeks, and Harry~L Swinney}, {\em Observation of
  anomalous diffusion and l{\'e}vy flights in a two-dimensional rotating flow},
  Physical Review Letters, 71 (1993), p.~3975.

\bibitem{solomon1994chaotic}
\leavevmode\vrule height 2pt depth -1.6pt width 23pt, {\em Chaotic advection in
  a two-dimensional flow: L{\'e}vy flights and anomalous diffusion}, Physica D:
  Nonlinear Phenomena, 76 (1994), pp.~70--84.

\bibitem{sugimoto1989generalized}
{\sc N~Sugimoto}, {\em Generalized burgers equations and fractional calculus},
  Nonlinear wave motion, 408 (1989), pp.~162--179.

\bibitem{sugimoto1991burgers}
\leavevmode\vrule height 2pt depth -1.6pt width 23pt, {\em Burgers equation
  with a fractional derivative; hereditary effects on nonlinear acoustic
  waves}, Journal of fluid mechanics, 225 (1991), pp.~631--653.

\bibitem{sweilam2014chebyshev}
{\sc NH~Sweilam, MM~Khader, and M~Adel}, {\em Chebyshev pseudo-spectral method
  for solving fractional advection-dispersion equation}, Applied Mathematics, 5
  (2014), p.~3240.

\bibitem{tadjeran2007second}
{\sc Charles Tadjeran and Mark~M Meerschaert}, {\em A second-order accurate
  numerical method for the two-dimensional fractional diffusion equation},
  Journal of Computational Physics, 220 (2007), pp.~813--823.

\bibitem{vafai2015handbook}
{\sc Kambiz Vafai}, {\em Handbook of porous media}, Crc Press, 2015.

\bibitem{zaslavsky1999physics}
{\sc Georg~M Zaslavsky and JD~Meiss}, {\em Physics of chaos in hamiltonian
  systems}, Nature, 398 (1999), p.~303.

\bibitem{zayernouri2015tempered}
{\sc Mohsen Zayernouri, Mark Ainsworth, and George~Em Karniadakis}, {\em
  Tempered fractional sturm--liouville eigenproblems}, SIAM Journal on
  Scientific Computing, 37 (2015), pp.~A1777--A1800.

\bibitem{zayernouri2015unified}
\leavevmode\vrule height 2pt depth -1.6pt width 23pt, {\em A unified
  petrov--galerkin spectral method for fractional pdes}, Computer Methods in
  Applied Mechanics and Engineering, 283 (2015), pp.~1545--1569.

\bibitem{zayernouri2013fractional}
{\sc Mohsen Zayernouri and George~Em Karniadakis}, {\em Fractional
  sturm--liouville eigen-problems: theory and numerical approximation}, Journal
  of Computational Physics, 252 (2013), pp.~495--517.

\bibitem{Zayernouri_FDDEs_2013}
{\sc M.~Zayernouri and G.~E. Karniadakis}, {\em Spectral and discontinuous
  spectral element methods for fractional delay differential equations},
  Submitted to SIAM J. Scientific Computing,  (2013).

\bibitem{zayernouri2014exponentially}
{\sc Mohsen Zayernouri and George~Em Karniadakis}, {\em Exponentially accurate
  spectral and spectral element methods for fractional odes}, Journal of
  Computational Physics, 257 (2014), pp.~460--480.

\bibitem{zayernouri2014fractional}
\leavevmode\vrule height 2pt depth -1.6pt width 23pt, {\em Fractional spectral
  collocation method}, SIAM Journal on Scientific Computing, 36 (2014),
  pp.~A40--A62.

\bibitem{zayernouri20}
\leavevmode\vrule height 2pt depth -1.6pt width 23pt, {\em Fractional spectral
  collocation methods for linear and nonlinear variable order fpdes}, Journal
  of Computational Physics, 293 (2015), pp.~312--338.

\bibitem{zayernouri2016fractional}
{\sc Mohsen Zayernouri and Anastasios Matzavinos}, {\em Fractional
  adams--bashforth/moulton methods: An application to the fractional
  keller--segel chemotaxis system}, Journal of Computational Physics, 317
  (2016), pp.~1--14.

\bibitem{zeng2015numerical}
{\sc Fanhai Zeng, Changpin Li, Fawang Liu, and Ian Turner}, {\em Numerical
  algorithms for time-fractional subdiffusion equation with second-order
  accuracy}, SIAM Journal on Scientific Computing, 37 (2015), pp.~A55--A78.

\bibitem{zeng2016fast}
{\sc Fanhai Zeng, Zhongqiang Zhang, and George~Em Karniadakis}, {\em Fast
  difference schemes for solving high-dimensional time-fractional subdiffusion
  equations}, Journal of Computational Physics, 307 (2016), pp.~15--33.

\bibitem{zhang2010galerkin}
{\sc H~Zhang, Fawang Liu, and Vo~Anh}, {\em Galerkin finite element
  approximation of symmetric space-fractional partial differential equations},
  Applied Mathematics and Computation, 217 (2010), pp.~2534--2545.

\bibitem{EPFL-ARTICLE-217895}
{\sc Lijing Zhao, Weihua Deng, and Jan~S. Hesthaven}, {\em Spectral methods for
  tempered fractional differential equations}, Mathematics of {C}omputation,
  (2016).

\bibitem{zhao2015second}
{\sc Xuan Zhao, Zhi-zhong Sun, and George~Em Karniadakis}, {\em Second-order
  approximations for variable order fractional derivatives: algorithms and
  applications}, Journal of Computational Physics, 293 (2015), pp.~184--200.

\end{thebibliography}

\end{document}